\documentclass[journal]{IEEEtran}
\usepackage{amsmath,amsfonts,amsthm}
\usepackage{algorithmic}
\usepackage{algorithm}
\usepackage{array}
\usepackage[caption=false,font=normalsize,labelfont=sf,textfont=sf]{subfig}
\usepackage{textcomp}
\usepackage{stfloats}
\usepackage{url}
\usepackage{verbatim}
\usepackage{graphicx}
\usepackage{cite}

\usepackage{bbm}
\usepackage{multirow}

\newcommand{\bfA}{{\bf A}}
\newcommand{\bfE}{{\bf E}}
\newcommand{\bfQ}{{\bf Q}}
\newcommand{\bfT}{{\bf T}}
\newcommand{\bfe}{{\bf e}}
\newcommand{\bfJ}{{\bf J}}
\newcommand{\bfs}{{{\bf s}}}
\newcommand{\bfx}{{\bf x}}
\newcommand{\bfy}{{\bf y}}
\newcommand{\bfz}{{\bf z}}
\newcommand{\bfr}{\hat{\bf r}}
\newcommand{\bfw}{{\bf w}}
\newcommand{\be}{\begin{equation}}
\newcommand{\ee}{\end{equation}}
\newcommand{\beq}{\begin{eqnarray}}
\newcommand{\eeq}{\end{eqnarray}}

\newtheorem{theorem}{Theorem}
\newtheorem*{theorem*}{Theorem}
\newcommand{\rev}[1]{\textcolor{black}{#1}}

\usepackage{color}

\begin{document}
\thispagestyle{empty}

\title{\Huge Open and Closed-Loop Weight Selection for Pattern Control of Paraboloidal Reflector Antennas with Reconfigurable Rim Scattering}

\author{R. M. Buehrer, {\it IEEE Fellow}, W. W. Howard, {\it Student Member, IEEE} and S. W. Ellingson, {\it Senior Member, IEEE}\thanks{The authors are with {\em Wireless @ Virginia Tech}, Bradley Dept. of Electrical and Computer Engineering, Virginia Tech, Blacksburg, VA 24060, USA.}\thanks{Corresponding author: buehrer@vt.edu}}

\maketitle

\begin{abstract}
\rev{It has been demonstrated that modifying the rim scattering of a paraboloidal reflector antenna through the use of reconfigurable elements along the rim facilitates sidelobe modification including cancelling sidelobes.  
In this work we investigate several open questions with respect to algorithms for determining the weights.  
First, we derive the general weight values needed at each reconfigurable element to place nulls at arbitrary angles.  
Second, since in many cases these weights require gains other than one, we develop a technique for determining unit-modulus weights so as to allow for surfaces which merely modify the phase of the scattered field,  while substantially reducing the gain at arbitrary angles.  Specifically, it is shown that despite the large search space (and non-convexity in the presence of discrete weights), weights can be found with reasonable computational complexity which provide useful cancellation capability.  Third, it is demonstrated that this can be done using open-loop (i.e., with pattern knowledge), closed-loop (without pattern knowledge), or hybrid (with inexact pattern knowledge) techniques both quantized and unquantized phases.   A primary finding is that sufficiently deep nulls are possible with essentially no change in the main lobe gain with practical (binary or quaternary) phase-only weights.}

\end{abstract}


\section{Introduction}
\label{sec:intro}
\rev{ Radio astronomy relies on large reflector antennas to receive weak signals \cite{rohlfs2013tools,CondonRansom+2016,ellingson2015antennas}.  However,   these applications are vulnerable to interference via sidelobes. One solution to interference is time and frequency blanking \cite{series2013techniques,Briggs05}.  This problem can be ameliorated without sacrificing data by sidelobe modiﬁcation or better yet,  sidelobe cancellation.  The latter involves placing a pattern null so as to reject the interference, ideally without impacting the main lobe gain. The traditional approach to sidelobe canceling is to use an array  of feeds \cite{bird2015fundamentals}.  The short-coming of such an approach is that it introduces greater aperture blockage and dynamic variability in the gain and shape of the main lobe.}

\rev{Recently, another approach for cancelling (or generally modifying) the sidelobe(s) of a reflector antenna pattern has been proposed  that  uses reconﬁgurable rim scattering \cite{Ellingson21}.  In that work it was shown that by using reconfigurable elements along the rim of a refletor antenna that introduce phase shifts to the reflected signal, sidelobes  can be altered and even cancelled.  That work examined  prime focus-fed circular axisymmetric paraboloidal reﬂectors,  although the concept is not limited to that type of system.  Specifically, it was shown in \cite{Ellingson21} that as long as  sufficient surface area along the rim is reconfigurable, sidelobe cancellation is possible.  Additionally, the work showed that cancellation did not  require continuously variable phase, but was in fact possible using quantized phase shifts  (viz., binary or quaternary).  While rim-loading had been previously proposed for antenna performance enhancement \cite{Bucci81}, to the best of our knowledge \cite{Ellingson21} is the first work  to propose sidelobe cancellation/modification through reconfigurable elements on the antenna rim.  A rudimentary example  algorithm for determining the binary weights was also described.  We shall refer to this algorithm as {\it serial search}.  }

In the current work we build on the idea proposed in \cite{Ellingson21} by describing algorithms for determining the complex weights needed at each reconfigurable element to cancel sidelobes at specific angles from the reflector axis using open-loop (initially described in \cite{Buehrer2022,Buehrer2023}), closed-loop, and hyrbid approaches.  We define ``open-loop'' approaches as algorithms which determine the appropriate element weights without using the received signal, but instead rely on knowledge of the antenna pattern and the contribution of each element to the pattern.  ``Closed-loop'' approaches are algorithms which determine the appropriate element weights by examining the received signal at the output of the feed (i.e., the total received signal collected by the reconfigurable dish), but without  knowledge (or with limited knowledge) of the antenna pattern. 

More specifically, in this paper we first derive the optimal weights needed to cancel sidelobes at an arbitrary angle $\psi$.\footnote{We will define the specific meaning of $\psi$ shortly.}  We then show that the approach can cancel sidelobes at multiple angles and describe the algorithm to determine the optimal element weights for this case.  However, the optimal weights prove to have two shortcomings:  (a) the complex values are, in general, not unit-modulus; and (b) the optimal weights cause the main lobe gain to vary as a function of the null direction.  

To overcome the latter problem, we show that a constraint for the main lobe gain can be added to the optimization which removes the main lobe gain variability.  The former problem (i.e., requiring element weights with non-unit gain) is less desirable since this would require elements with the ability to attenuate or amplify the signal while scattering.  This is clearly undesirable from a cost/complexity perspective.  To overcome this problem, we formulate a least-squares problem and utilize a version of the gradient projection algorithm to find weights that are constrained to have unit modulus.  It is shown that this algorithm can be successfully used to find weights that form one or more nulls while maintaining  constant main lobe gain. 

While the algorithm described above achieves what we desire, it requires continuously-variable phase at each reconfigurable element.  To overcome this requirement, we must quantize the weights at each element to a small number of possible phase shifts.  Unfortunately, this now becomes an optimization problem with a discrete multi-dimensional search space.  The non-convex discrete nature of this optimization problem does not lend itself to the aforementioned algorithm.  Thus, we employ a different approach.  Specifically we apply the Firefly Algorithm and show that this results in weights that yield the desired nulls.  A primary finding is that sufficiently deep nulls are possible with negligible change in the main lobe with practical (binary or quaternary) phase-only weights.

Using the algorithms described above, we demonstrate that it is possible to find the necessary weights with reasonable complexity and a practical (i.e., limited) number of states.  However, these ``open-loop'' algorithms require that the antenna pattern is known.  Given that the pattern is typically known only approximately, or may be dynamically varying due to e.g., gravitational deformation, it is of interest to find algorithms that allow the weights to be determined without knowledge of the pattern.  Thus, we also develop techniques to determine the weights without such knowledge using a ``closed-loop'' approach that uses the received interfering signal in the output of the feed to drive weight adaptation.  We examine two approaches: (1) a power minimization approach that assumes that the received signal is dominated by the interfering signal; and (2) a library-based method that relies on a library of weights determined from a calibration process.  Both approaches demonstrated success, but require longer convergence times.  To address the convergence time, we also develop a hybrid approach that uses partial pattern knowledge to speed up convergence.

The work in this paper incorporates and builds on our previous work \cite{Ellingson21,Buehrer2022,Buehrer2023} in several ways.  First, in the context of open-loop approaches, we extend the work in  \cite{Buehrer2022,Buehrer2023,Ellingson21} by examining the frequency sensitivity of the optimal weights.  We also explore two approaches to handle frequency sensitivity by either (a) recalculating the optimal weights for a new frequency starting from the previous weights or (b) creating optimization criteria at multiple frequencies.   In addition to this extension of the open-loop approach, we also develop two closed-loop approaches.  The first approach drives the weight adaptation using output power.  For this approach we derive the probability of decision error based on the interference-to-noise ratio, the antenna gain in the direction of interest and the change in gain due to a change in the weights.  Because of the potentially long convergence time of this approach, we also explore the structure of the final weights and show that clustering the weights can improve convergence time.  We also develop a second closed-loop approach based on a library of weights.  This approach also allows for faster convergence.  For both approaches we examine the performance in the presence of a moving source.  Finally, we also develop a hybrid approach that uses the open-loop solution as a starting point for closed-loop iterations.  We show that the approach can potentially improve convergence time by an order of magnitude.

This paper is organized as follows:  Section \ref{sec:system_model} describes the system model used in this work.  The algorithms for determining the weights for each reconfigurable element in an open-loop fashion are presented in \ref{sec:weights}, Section \ref{sec:results} presents numerical results for these open-loop approaches.  The closed-loop approach is described in Section \ref{sec:closedloop} including a real-time simulation example with a single interferer.  A hybrid approach is examine in Section \ref{sec:hybrid}.  Conclusions are presented in Section \ref{sec:concl}.

\section{System Model}
\label{sec:system_model}
The antenna system assumed in this paper is presented in Figure \ref{fig:fig0}.  Following the development in \cite{Ellingson21} we assume the equivalence of transmit and receive patterns and calculate the transmit patterns using physical optics (PO).  The total electric field intensity $\bfE^s$ scattered by the reflector in the far-field direction $\psi$ is given by\footnote{For simplicity/clarity and with no loss of generality we will examine the antenna pattern in the H-plane.  Further, we examine the pattern versus an angle we define as $\psi$ which is the angle from the main lobe in the H-plane (i.e., $\psi=\theta$ with $\phi=90^\circ$).  Clearly the antenna pattern in general depends on both angular coordinates, but for ease of exposition we choose to stay in the H-plane in this discussion.} 
\be
\bfE^s(\psi) = \bfE^s_f(\psi) + \bfE_r^s(\psi)
\ee
where $\bfE_f^s$ is the electric field intensity due to the fixed portion of the dish
\be
\label{eq:field_fixed}
\bfE_f^s(\psi) = -j\omega \mu_o \frac{e^{-j\beta r}}{4\pi r} \int_{\theta_f=0}^{\theta_1} \int_{\phi=0}^{2\pi} \bfJ_0(\bfs^i) e^{j\beta\bfr(\psi)\cdot\bfs^i}ds,
\ee
 where $\bfJ_0(\bfs^i)$ is the PO equivalent surface current distribution, $j=\sqrt{-1}$, $\omega=2\pi f$ is the operating frequency in rad/s, $\mu_0$ is the permeability of free space, $\beta$ is the wavenumber, $\bfr(\psi)$ points from the origin towards the field point in the direction $\psi$, $\theta_f$ is the angle measured from the reﬂector axis of rotation toward the rim with $\theta_f = \theta_1$ at the rim of the fixed portion of the dish and $\theta_f=\theta_0$ at the rim of the entire dish (see Figure \ref{fig:fig0}), $\phi$ is the angular coordinate orthogonal to both $\theta_f$ and the reﬂector axis, and $ds$ is the differential element of surface area.  See Figure \ref{fig:fig0} for more details.
\begin{figure}
    \centering
    \vspace{1cm}
    \includegraphics[ width=0.75\columnwidth]{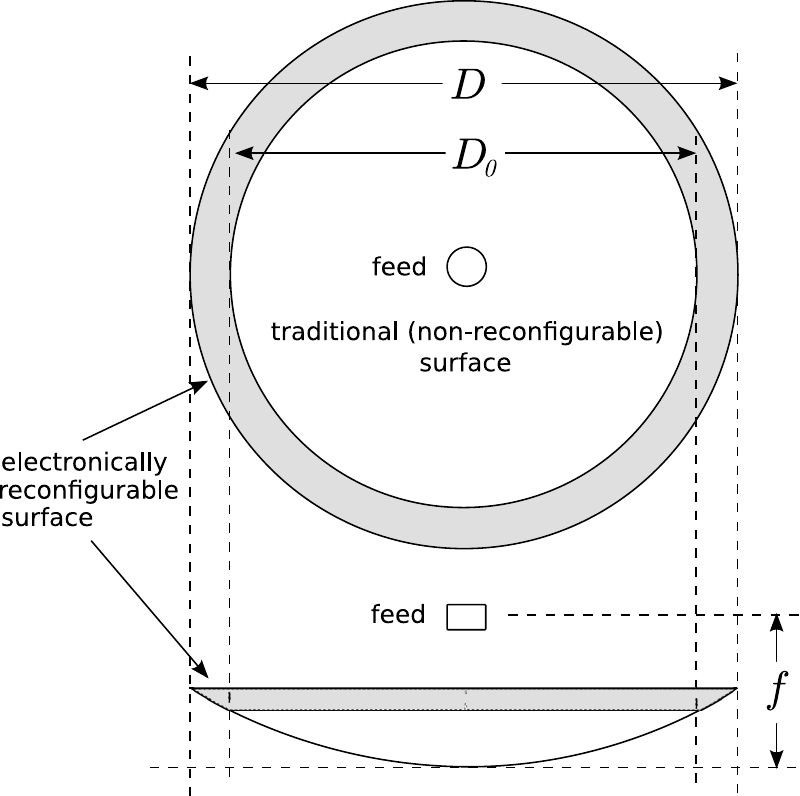}
    \includegraphics[ width=0.5\columnwidth]{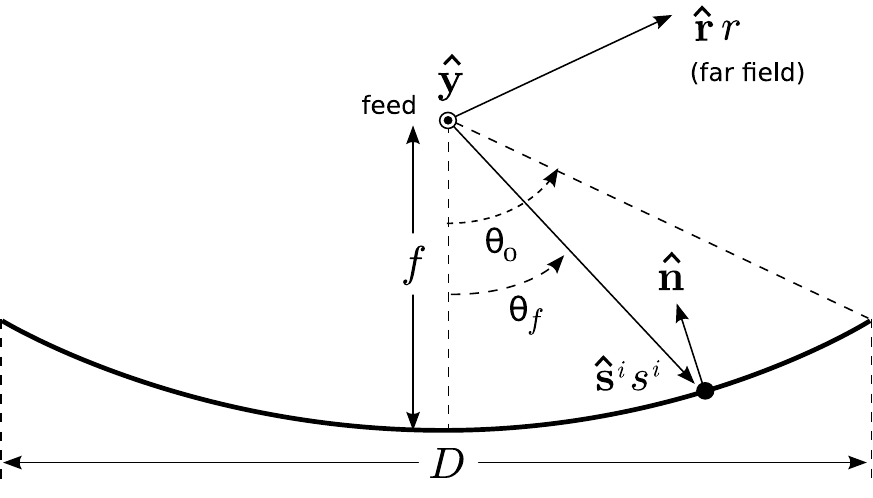}
    \caption{On-axis (top) and side (middle) views of an electronically-reconﬁgurable rim scattering system along with the geometry for analysis (bottom) assumed in this paper.}
    \label{fig:fig0}
    \vspace{1cm}
\end{figure}

The electric field intensity due to the reconfigurable portion of the dish is similarly written as
\be
\bfE_r^s(\psi) = -j\omega \mu_0 \frac{e^{-j\beta r}}{4\pi r} \int_{\theta_f=\theta_1}^{\theta_0} \int_{\phi=0}^{2\pi} \bfJ_1(\bfs^i) e^{j\beta \bfr(\psi)\cdot\bfs^i}ds
\ee
where the major differences between the contributions due to the fixed and reconfigurable portions of the dish are (a) the angles of integration and (b) the PO equivalent surface current distribution.  Due to the discrete nature of the reconfigurable surface, we can write $\bfE_r^s(\psi)$ as 
\be
\label{eq:field_reconfig}
\bfE_r^s(\psi) = -j\omega \mu_o \frac{e^{-j\beta r}}{4\pi r} \sum_n \bfJ_1(\bfs_n^i) e^{j\beta\bfr(\psi)\cdot\bfs^i_n}\Delta s
\ee
where $\bfJ_1(\bfs_n^i)=w_n \bfJ_0(\bfs_n^i) $ is the current distribution due to the $n$th element with complex-valued weight $w_n$.  These weights will be designed to cancel sidelobes in the H-plane co-pol pattern.  Thus, we are primarily concerned with the $y$-component of the vector $\bfE_r^s(\psi)$ and define the complex scalar $E_r^{s,co}(\psi)$ to be the $y$-component of the vector $\bfE_r^s(\psi)$.  Also, for convenience, we can write $E_r^{s,co}(\psi)$  in terms of two $N \times 1$ dimensional arrays $\bfe_\psi$ and $\bfw$ representing the co-pol portion of the electric field intensity without the influence of the elements and the complex-valued element gains respectively:
\be
E_r^{s,co}(\psi) = \bfe_\psi^T \bfw
\ee
where $\bfx^T$ is the transpose of the array $\bfx$, $w_n$ is the complex-valued weight applied to the $n$th reconfigurable element and the $n$th element of $\bfe_\psi$ is 
\be
 e_{\psi,n} = \left (-j\omega \mu_o \frac{e^{-j\beta r}}{4\pi r} \bfJ_o(\bfs_n^i) e^{j\beta \bfr(\psi)\cdot\bfs^i_n}\Delta s \right ) \cdot \hat{\bf y},
\ee
 where the dot product $\left (\bfe \right )\cdot \hat{\bf y}$ selects the $y$-component of the vector $\bfe$.  Note that both one-dimensional arrays are of length $N$ where $N$ is equal to the number of reconfigurable elements placed along the rim of the dish. 

Now to cancel the sidelobe gain at angle $\psi$, we wish 
\be
E_f^{s,co}(\psi) + E_r^{s,co}(\psi) = 0
\ee
or in other words, $E_r^{s,co}(\psi) = -E_f^{s,co}(\psi)$.  Thus, we wish to find the set of weights  $\bfw$ such that
\be
\bfe_\psi^T\bfw = -E_f^{s,co}(\psi)
\label{eq:req}
\ee
In the following section we describe techniques to determine $\bfw$.

\section{Open-Loop Weight Selection} \label{sec:weights}
The approach to find the appropriate weights to satisfy (\ref{eq:req}) depends on the restrictions placed on $\bfw$.  First, we describe how to determine the weights with no restrictions placed on $\bfw$.  Second, we will describe techniques for finding $\bfw$ if the weights are restricted to unit-modulus values (i.e., phase-shifting only).  Finally, we describe approaches to finding weights if we further restrict the weights to be discrete (i.e., come from a finite set).

\subsection{Unquantized Weights}
While (\ref{eq:req}) provides the requirement for creating a null at angle $\psi$, it does not provide the specific weights $\bfw$ to achieve this.
If the weights are unconstrained, to form a null at angle $\psi$ we can simply let $\bfw$ be equal to
\be
\bfw_{opt} = -E_f^{s,co}(\psi) \frac{\bfe_\psi^*}{||\bfe_\psi||^2_2}
\label{eq:opt}
\ee
where $\bfx^*$ represents a vector where each element is the complex conjugate of the corresponding element in $\bfx$ and $||\bfx||_2$ represents the 2-norm of the vector $\bfx$.  We will term these {\it optimal} weights since these weights guarantee that the response of the dish in the direction $\psi$ is zero.  Unfortunately, $\bfw_{opt}$ in general has elements with $|w_i|\neq 1$ which requires that the elements have controllable gains (i.e., can provide attenuation or gain to the scattered field). Such a requirement is undesirable from a cost and complexity perspective.  Thus, we seek to restrict the weights such that $|w_i|=1$.  This can be written as the following minimization problem
\begin{eqnarray} \label{eq:cost_unq}
     \bfw_{gp} = \min_{\bfw \in {\cal C}^N} & \left \| E_f^{s,co}(\psi) + \bfe_\psi^T\bfw  \right \|_2^2 & \label{eq:opt_single} \\
    &  s.t. \hspace{0.25cm} |w_i| = 1 & i =1, 2, \ldots N \nonumber 
\end{eqnarray}
This is a non-convex, complex-valued, constant-modulus, least squares optimization problem which is a special case of non-convex quadratically-constrained quadratic programming \cite{Luo10}.  While the problem is non-convex in general, simulation results suggest that the cost function defined in (\ref{eq:cost_unq}) has a number of good local minima. One solution to this problem is to use semi-definite programming \cite{Luo10}.  Unfortunately, this expands the problem such that it is of dimension $N^2$.  Given that for moderate to large antennas ($D=15$ m to $D=100$ m) and GHz frequencies, $N$  can range from $10^3$ to $10^5$, such an expansion is computationally expensive.  A more efficient solution is the use of {\it Gradient Projection} \cite{Tranter17}.  Gradient Projection (GP) uses standard gradient descent followed by a projection onto a convex set.  While our set of unit-modulus vectors is not a convex set, this approach has been shown to converge for projection onto the set of unit-modulus vectors \cite{Tranter17}.  Adapting GP to our problem, we have the following algorithm.  At the $k$th iteration we apply standard gradient descent to create a new vector 
\be
\tilde{\bfw}^{(k+1)} = \bfw^{(k)} + \alpha \bfe^* \left ( E_f^{s,co}(\psi) + \bfe^T\bfw^{(k)} \right )
\ee
where $\alpha = \frac{\gamma}{\lambda_{max}\left ( \bfe^* \bfe^T \right) }$ is the adaptation constant and $\gamma \in (0,1)$.  Of course, this is not guaranteed to provide unit-modulus weights.  Thus, we follow this with a projection step which projects $\tilde{\bfw}$ onto an array of weights with unit-modulus $\bfw$:
\be
\bfw^{(k+1)} = {\cal P} \left (\tilde{\bfw}^{(k+1)} \right )
\ee
where ${\cal P}(\bfx) = e^{j\angle \bfx}$ and the vector operator $\bfz =\angle \bfx$ results in a vector $\bfz$ where the $i$th element is the angle of $x_i$, i.e.,  $z_i=\angle x_i$ .  The detailed algorithm follows.

\begin{table}[h]
    \centering
    \begin{tabular}{l|l}
    \hline 
    \multicolumn{2}{c}{GP Algorithm (single constraint/null)} \\
    \hline
     1. & Initialization: $k=0$, $\alpha = \frac{\gamma}{\lambda_{max}\left ( \bfe^* \bfe^T \right) }$, $\gamma \in (0,1)$ \\
      & \hspace{2cm} $\bfw^{(0)} =e^{j\angle \left ((\bfe^*\bfe^T)^{-1}\bfe^*(-E_f^{s,co}(\psi)) \right )} $ \\
      & \\
         & Repeat \\
      2. & \hspace{1cm} ${\bf \eta}^{(k+1)} = \bfw^{(k)} - \alpha \bfe^* \left ( E_f^{s,co}(\psi) + \bfe^T\bfw^{(k)} \right )$ \\
      3. & \hspace{1cm} $\bfw^{(k+1)} = e^{j\angle {\bf \eta}^{(k+1)}}$ \\
      4. & \hspace{1cm} $k=k+1$ \\
         & Until convergence  \\
      \hline
    \end{tabular}
    \label{tab:my_label}
\end{table}

\subsection{Main Lobe Variation and Nulling Multiple Angles}
In the above formulation, we chose weights to cancel the sidelobe at a single angle.  However, there may be multiple angles that we wish to null. Further, as we will show later, nulling a particular angle may cause small changes to the main lobe gain (this was also shown in \cite{Ellingson21}).  This variation  may be tolerable in some applications, but in radio astronomy it can be a significant problem\cite{CondonRansom+2016}.  Thus, we wish to constrain this variation as much as possible.  To accomplish both goals (adding nulls and constraining main lobe variation) we simply need to incorporate these requirements in the cost function.  

For example, consider the case of $K$ desired nulls.  In this case we wish
\beq
E_f^{s,co}(\psi_1) + E_r^{s,co}(\psi_1) & = & 0 \nonumber \\ 
E_f^{s,co}(\psi_2) + E_r^{s,co}(\psi_2) & = & 0 \nonumber \\ 
& \vdots & \nonumber \\
E_f^{s,co}(\psi_K) + E_r^{s,co}(\psi_K) & = & 0 
\eeq
Or in other words we desire
\beq
 \bfe_{\psi_1}^T\bfw & = & -E_f^{s,co}(\psi_1) \nonumber \\ 
& \vdots & \nonumber \\
 \bfe_{\psi_K}^T\bfw & = & -E_f^{s,co}(\psi_K)   
\eeq
Defining 
\beq
\bfQ & = & \left [\bfe_{\psi_1}, \bfe_{\psi_2}, \ldots \bfe_{\psi_K}   \right ] \\ 
\bfy & = & -[E_f^{s,co}(\psi_1), E_f^{s,co}(\psi_2), \ldots E_f^{s,co}(\psi_K)  ]^T, 
\eeq
we have the requirement 
\be
\bfQ^T \bfw = \bfy
\ee
Using the least squares solution we obtain the optimal weights
\be
\bfw_{opt} = \bfQ \left ( \bfQ^T \bfQ\right )^{-1} \bfy
\ee

To avoid variation in the main lobe, we can include a constraint for the main lobe.  Specifically, for $\psi=0$ we require
\be
\bfe_{0}^T\bfw = \kappa
\ee
where $\kappa$ is the target constraint.  Ideally, we could set $\kappa = \bfe_{0}^T{\bf 1}_N$ where ${\bf 1}_N$ is an $N\times 1$ array of all ones which would provide the same main lobe gain as the fixed reflector.  With unconstrained weights, this is possible.  However, with weights restricted to unit modulus, satisfying both the main lobe constraint and side lobe constraints many not be possible.  One option is to ease the constraints by choosing $\kappa=0$, although this removes any contribution of the reconfigurable portion of the reflector to the main lobe.  Alternatively, we could choose $\kappa = \delta E_f^{s,co}(0)$ for some small value $\delta$ to provide some additional gain in the main lobe.  Experimentally we have found this latter approach to be successful. Thus, we can modify the vector of requirements to be
\be
\bfy = \left [\kappa,-E_f^{s,co}(\psi_1), -E_f^{s,co}(\psi_2), \ldots -E_f^{s,co}(\psi_K) \right ]^T
\ee
where $\psi_1, \psi_2, \ldots \psi_K$ are the $K$ angles at which we desired to place nulls and the first element of $\bfy$  corresponds to the constraint placed on the change to the main lobe gain.  Further, we define the matrix
\be
\bfA  = \left [ \bfe_{\psi_0},\bfe_{\psi_1}, \bfe_{\psi_2},\ldots \bfe_{\psi_K},\right ]^T
\ee
The least squares problem then becomes:
\begin{eqnarray}
   \bfw_{gp} =  \min_{\bfw \in {\cal C}^N} & \left \| \bfy + \bfA \bfw  \right \|_2^2 & \label{eq:opt_multiple} \\
    &  s.t. \hspace{0.25cm} |w_i| = 1 & i =1, 2, \ldots N \nonumber 
\end{eqnarray}

The resulting algorithm is then

\begin{table}[h]
    \centering
    \begin{tabular}{l|l}
    \hline 
    \multicolumn{2}{c}{GP Algorithm (multiple constraints)} \\
    \hline
     1. & Initialization: $k=0$, $\alpha = \frac{\gamma}{\lambda_{max}\left ( \bfA^H \bfA \right) }$, $\gamma \in (0,1)$ \\
      & $\bfw^{(0)} =e^{j\angle \left ((\bfA^H\bfA)^{-1}\bfA^H\bfy \right )} $ \\
      & \\
       & Repeat \\
      3. & \hspace{1cm} ${\bf \eta}^{(k+1)} = \bfw^{(k)} - \alpha \bfA^H \left ( \bfy + \bfA \bfw^{(k)} \right )$ \\
      4. & \hspace{1cm}  $\bfw^{(k+1)} = e^{j\angle {\bf \eta}^{(k+1)}}$ \\
      5. & \hspace{1cm}  $k=k+1$ \\
        & Until convergence  \\
      \hline
    \end{tabular}
    \label{tab:my_label2}
\end{table}
\subsection{Quantized Weights}
The weights described in the previous sections have one primary disadvantage: they presume continuously-variable phase values.  In a practical implementation, it is much more reasonable that only a finite number of phase values would be available on each reconfigurable element.  Thus, we wish to solve the single-null problem in (\ref{eq:opt_single}) or the multiple-null problem (\ref{eq:opt_multiple}) with the additional constraint that 
\be
w_i \in {\cal W}=\left \{ e^{j2\pi/M},e^{j4\pi/M}, \ldots e^{j2\pi(M-1)/M} \right \}.
\ee
where $M$ is the number of possible phase values.  The most straightforward approach is to simply search over all  vectors $\bfw\in {\cal W}$ to find the one that minimizes the cost function in (\ref{eq:opt_single}) for a single constraint or (\ref{eq:opt_multiple}) for multiple constraints.  Unfortunately, this requires a search over $M^N$ possible vectors where (as discussed above) $N$ can be on the order of $10^3$ to $10^5$ for moderate to large antennas and GHz frequencies.   For example, in the particular case we will examine $N=2756$ (see Section \ref{sec:results} and \cite{Ellingson21} for details). Thus, even using binary weights a brute-force approach results in a search over $2^{2756}$ possible values of $\bfw$ which is obviously infeasible.  

\rev{ As a result, it is clear that a meta-heuristic approach is required.  A number of such approaches have been developed that use some form of randomness in the search.  In particular, nature-inspired approaches \cite{yang10} such as evolutionary approaches \cite{Lundardi18}, simulated annealing \cite{russell2002artificial}, ant colony approaches \cite{dorigo1999ant}, swarm optimization \cite{chen2009novel},  and firefly algorithms \cite{Li15,yang10,Lundardi18,FISTER201334,sayadi10} have demonstrated good performance.  In our own investigation, we found that the Firefly Algorithm provided the best trade-off between computational complexity and performance.  Firefly algorithms imitate the behavior of fireflies where each firefly moves towards other fireflies with higher intensity.  The attractiveness of one firefly to another depends both on the stronger firefly's intensity as well as the distance between them.  Specifically, we adapt the Firefly Algorithm as described below.}



\begin{table}[h]
    \centering
\rev{
\begin{tabular}{l|l}
    \hline 
    \multicolumn{2}{c}{Firefly Algorithm} \\
    \hline
     1. & Initialization: $N_{pop}$, $M$, $\gamma$ \\
      & \hspace{1.5cm} $\bfw_i^{(0)}$ randomly chosen from ${\cal W}^N$  \\
      & \hspace{1.75cm}  $\forall i\in \{1, \ldots N_{pop}\}$ \\
      2. & Calculate $I_i$ for $i=1, \ldots N_{pop}$ \\
      3. & Rank fireflies such that $I_i>I_{i-1} \forall i$ \\
       4. & for $iter=1:M$ \\
      5.   & \hspace{0.25cm} for $i=1:N_{pop}-1$\\
      6. & \hspace{0.5cm} for $j=i+1:N_{pop}$\\
      7. & \hspace{0.75cm} if $I_j>I_i$ \\
      8. & \hspace{1cm} Calculate distance between fireflies $d_H(i,j)$ \\
      9. & \hspace{1cm} Calculate attractiveness $\beta = \frac{1}{1+\gamma d_H(i,j)}$ \\
      10. & \hspace{1cm} for all elements where $w_j(k) \neq w_i(k)$  \\
         & $\hspace{1.5cm}$ with probability $p=\beta$, $w_j(k) = w_i(k)$   \\
        & \hspace{0.75cm} endif  \\
        & \hspace{0.5cm} end for ($j$) \\ 
        & \hspace{0.25cm} end for ($i$) \\ 
      11. & \hspace{0.25cm} Randomly move brightest firefly \\
          & \hspace{0.5cm} (keep movement if intensity increases) \\
     12.  & \hspace{0.25cm} Calculate intensity $I_i$ for $i=1, \ldots N_{pop}$ \\
      13.  & \hspace{0.25cm} Rank fireflies such that $I_i>I_{i-1} \forall i$   \\
       & end for (iter) \\
      \hline
    \end{tabular}
    }
    \label{tab:my_label3}
\end{table}

\rev{ The Firefly Algorithm begins by defining several parameters including the firefly population $N_{pop}$, the number of total movements (i.e., iterations) $M$, and the intensity parameter $\gamma$.  The initial population of fireflies have positions (i.e., weight vectors ${\bf w}_i$, $i=1\ldots N_{pop}$) that are randomly chosen from the set ${\cal W}^N$.  For each firefly an intensity value is calculated where the intensity is the inverse of the cost function, $I_i =1 /C({\bfw_i})$ where $C \left ( \bfw_i \right )=\left \| E_f^{s,co}(\psi) + \bfe_\psi^T\bfw_i  \right \|_2^2$ for a single constraint and  $ C \left ( \bfw_i \right )=\left \| \bfy + \bfA\bfw_i  \right \|_2^2$ for multiple constraints.  Fireflies are then ranked according to increasing intensity.  For each movement (iteration), fireflies are moved starting with the firefly with the lowest intensity. Each firefly moves towards fireflies with higher intensity.  The amount of movement is inversely related to the Hamming distance between the two fireflies. Specifically, a probability value $\beta$ is calculated that is inversely related to the Hamming distance.  With probability $\beta$, the firefly modifies each weight to be equal to the higher intensity firefly.  After all fireflies have been moved, the highest intensity firefly moves randomly, maintaining that movement only if it results in a higher intensity. } 

\rev{Overall, the complexity of the algorithm is 
\be
{\cal O}\left ( M \left ( \frac{N_{pop}(N_{pop}+1)}{2} + N_{pop} \log_2(N_{pop}\right ) \right) 
\ee
where the $M$ is the number of movements, first term related to the population size is due to the movement of each firefly in the direction of the higher intensity fireflies and the second term related to the population size is due to sorting of intensity values.  Generally, increasing $N_{pop}$ and $M$ improves performance, but also increases complexity.  In the results shown below, we set the population size to $N_{pop}=50$ and the number of movements to $M=600$ as these values provided a reasonable trade-off. }


\subsection{Results}\label{sec:results}
To demonstrate the performance of the above algorithms, in this section we provide numerical results.  We will assume a $D = 18$ m paraboloidal reflector operating at 1.5 GHz where the outer
0.5 m of the reﬂector surface consists of 2756 contiguous reconfigurable elements.  Each element is a square ﬂat plate conformal to the  paraboloidal surface having side length of approximately $0.5\lambda$ ({\it i.e.}, an area $\Delta s \approx 0. 25\lambda^2$).  The feed is modeled as a $\hat{y}$-oriented electrically-short electric dipole with field additionally modiﬁed by the factor $(\cos \theta_f )^q$ ,where $q$ controls the directivity of the feed.  Setting $q = 1.14$ yields edge illumination (i.e., ratio of ﬁeld intensity in the direction of the rim to the ﬁeld intensity in the direction of the vertex), to approximately $-11$ dB, yielding aperture efﬁciency of about $81.5$\%.  Figure \ref{fig:fig1} presents the H-plane co-pol pattern of an 18 m fixed reflector (i.e., without the reconfigurable elements or equivalently with the reconfigurable elements set to $w_i=1, \forall i$).  Also plotted is the same pattern for the reconfigurable reflector with the optimal weights (i.e., infinite quantization, and no unit modulus constraint) defined in (\ref{eq:opt}).  The weights are set to place a null at $\psi = 1.25^\circ$ which is directly on the peak of the first sidelobe in the fixed reflector pattern.  We can see from the figure that the sidelobe at $\psi=1.25^\circ$ is indeed cancelled.  However, the main lobe gain is reduced from 48.1 dBi (for the fixed 18m dish) to 47.7 dBi for the reconfigurable dish.   

The pattern in Figure \ref{fig:fig1} resulted from weights generated to place a null at $\psi=1.25^\circ$.  As a result, the gain at $\psi=1.25^\circ$ is essentially zero. This level of cancellation can be accomplished at any angle outside the main beam as shown in Figure \ref{fig:optimal_null}. Specifically, in Figure \ref{fig:optimal_null} we plot the H-plane co-pol gain achieved at $\psi$ (we will refer to this as $G(\psi)$) when the weights are generated using (\ref{eq:opt}) to place a null at $1^\circ\leq \psi \leq 3^\circ$ (see plot labeled ``optimal'').  Note that this is not a pattern, but rather the resulting null depth when attempting to place a null at $\psi$.    It can be seen that the optimal weights provide a gain of zero to within machine precision  for any angle between $1^\circ$ and $3^\circ$,  

Similarly, the resulting main lobe gain when cancelling a single sidelobe at angle $\psi$ is shown in Figure \ref{fig:main_lobe}.  It can be seen that the main lobe gain can vary by over 0.1 dB when cancelling with the optimal weights.  As discussed earlier, this variation can be problematic.  

\begin{figure}
    \centering
    \vspace{1cm}
    \includegraphics[trim={5cm 8cm 5cm 10cm}, width=0.75\columnwidth]{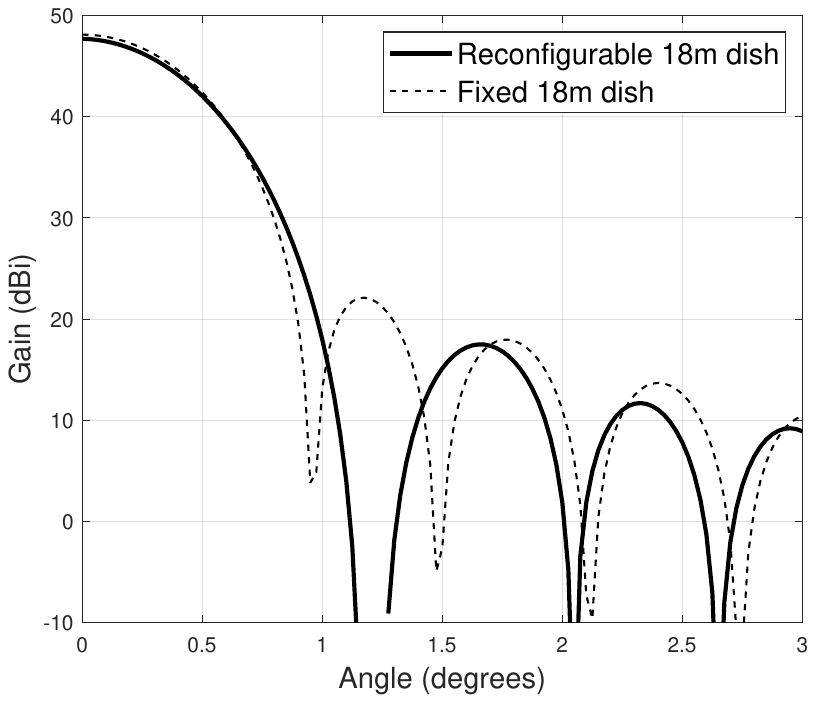}
    \caption{H-plane co-pol pattern for traditional (fixed) 18m dish and reconfigurable 18m dish with 0.5m reconfigurable rim ($\psi=1.25^o$)}
    \label{fig:fig1}
    \vspace{1cm}
\end{figure}

\begin{figure}
    \centering
    \includegraphics[trim={5cm 8cm 5cm 10cm}, width=0.75\columnwidth]{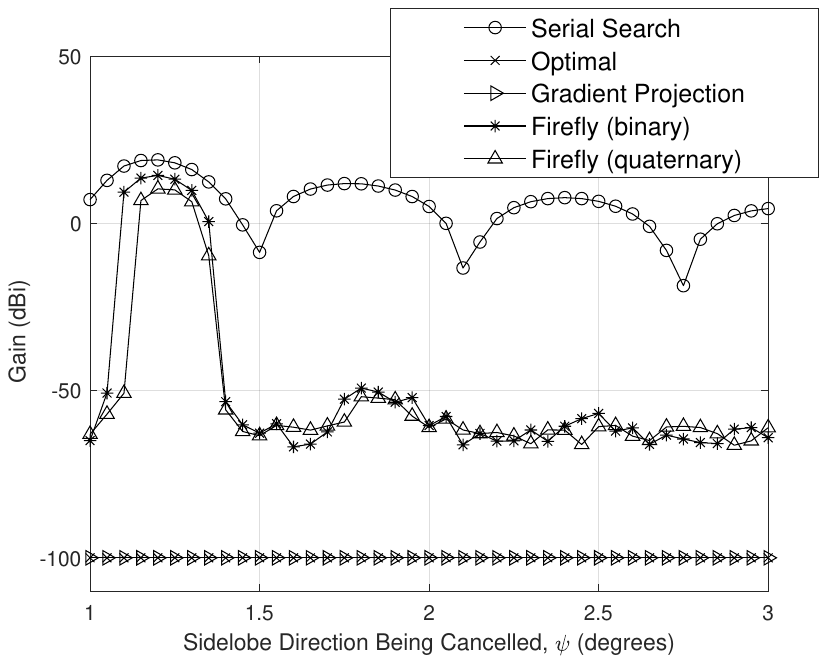}
    \caption{\rev{H-plane co-pol pattern gain $G(\psi)$ when placing a null at angle $\psi$ using the proposed algorithms and the serial search approach for binary weights in \cite{Ellingson21}.  Note that the gains for ``optimal'' and ``Gradient Projection'' were measured as zero to the within the limits of machine precision.  Thus, to visualize these values, we simply mapped them to $-100$ dBi for plotting purposes.} }
    \label{fig:optimal_null}
\end{figure}

Further, and as also discussed earlier, the optimal weights as determined by (\ref{eq:opt}) are not unit modulus and thus would require reconfigurable elements with controllable gain.  The gain achieved at the null direction when using the weights defined by the least-squares problem described in (\ref{eq:opt_single}) and determined using the Gradient Projection algorithm are also plotted in Figure   \ref{fig:optimal_null} (weights labeled ``Gradient Projection'').  It can be seen that the unit-modulus weights also achieve a gain of essentially zero. The corresponding gain is also plotted in Figure  \ref{fig:main_lobe}.  The additional constraint of the unit modulus causes more mainlobe variation as the angle of sidelobe cancellation ($\psi$) is changed.  Specifically, the main lobe gain will vary by approximately 0.3 dB  $\approx 7\%$ which may be significant in radio astronomy applications.  We will address this limitation shortly.  Of course, these weights, while unit modulus, are not quantized and thus require the elements to have continuously-variable phase.  

\rev{ The null depth achievable when using binary ($M=2$) weights (via the serial search algorithm described in \cite{Ellingson21} and the Firefly Algorithm) or quaternary weights (via the Firefly Algorithm)} are also shown in Figure \ref{fig:optimal_null}.  When using the serial search approach from \cite{Ellingson21}, the binary weights appear to get stuck in a local minimum that limits the null achievable\footnote{The exact reason for the performance is not known at this time.  Slightly different results are typically achieved if the search is started at different places or computed in a different order.}.  The minimum gains are in the range of 0-5 dBi which is a reduction in fixed pattern gain of only  5-15 dB. On the other hand, \rev{ when using the Firefly Algorithm, the quantized weights can push the sidelobe gains down to approximately $-50$ dBi, with the exception of a $0.2^\circ$ range around the maximum sidelobe peak at $\psi=1.25^\circ$.}  In that range, there is insufficient granularity in the weights to completely null the power from the fixed portion of the dish.  On the other hand, the main lobe gain has roughly 0.2 dB of variation when using Firefly and 0.3 dB of variation when using the serial search, although the latter provides a higher main lobe gain by nearly 0.5 dB.  Note that this difference in main lobe gain can be reduced by changing the main lobe constraint.

\begin{figure}
    \centering
    \vspace{1cm}
    \includegraphics[trim={5cm 8cm 5cm 10cm}, width=0.75\columnwidth]{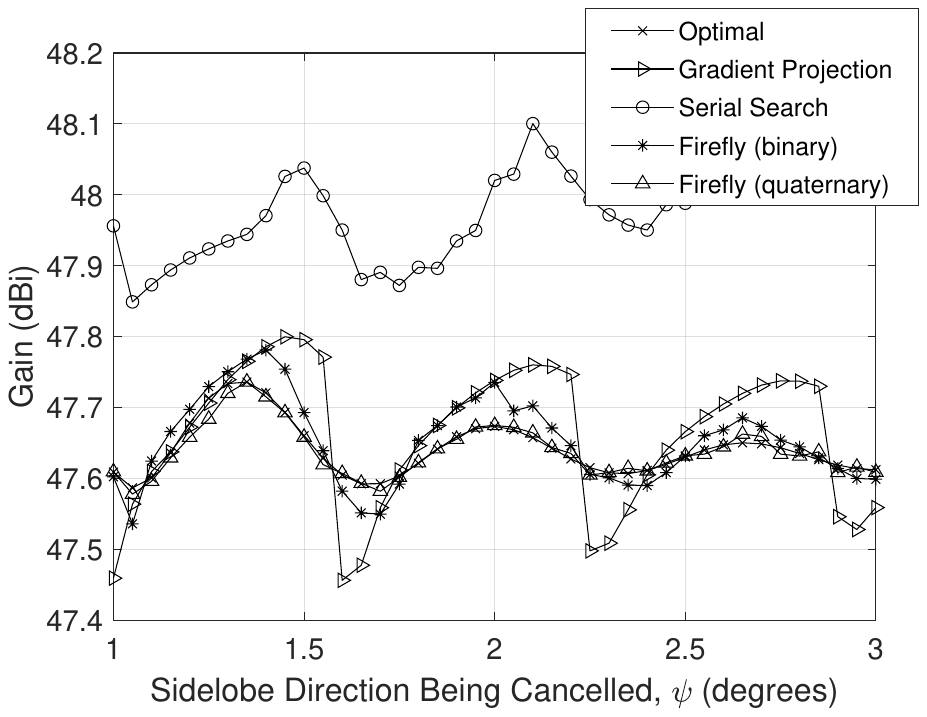}
    \caption{\rev{Main lobe gain for H-plane co-pol pattern ($G(0)$) when cancelling sidelobe at $\psi$ using the proposed algorithms and the serial search approach for binary weights in \cite{Ellingson21}}}
    \label{fig:main_lobe}
\end{figure}

To demonstrate the performance of the algorithms to create multiple nulls, we first examine the continuously-variable phase case using gradient projection to solve the least squares problem of (\ref{eq:opt_multiple}).  As an example, we applied the algorithm to three angles: $\psi=1.0^o$, $\psi=2.0^o$ and $\psi=3.0^o$.  The resulting pattern is plotted in Figure \ref{fig:fig4} along with the pattern of an 18m dish with a fixed pattern.  We can see that the reconfigurable rim provides nulls at all three angles, as desired.  There is a small loss in main lobe gain ($\sim$0.1 dB), but the two nulls are sufficiently deep.

\begin{figure}
    \centering
    \vspace{1cm}
    \includegraphics[trim={5cm 8cm 5cm 10cm}, width=0.75\columnwidth]{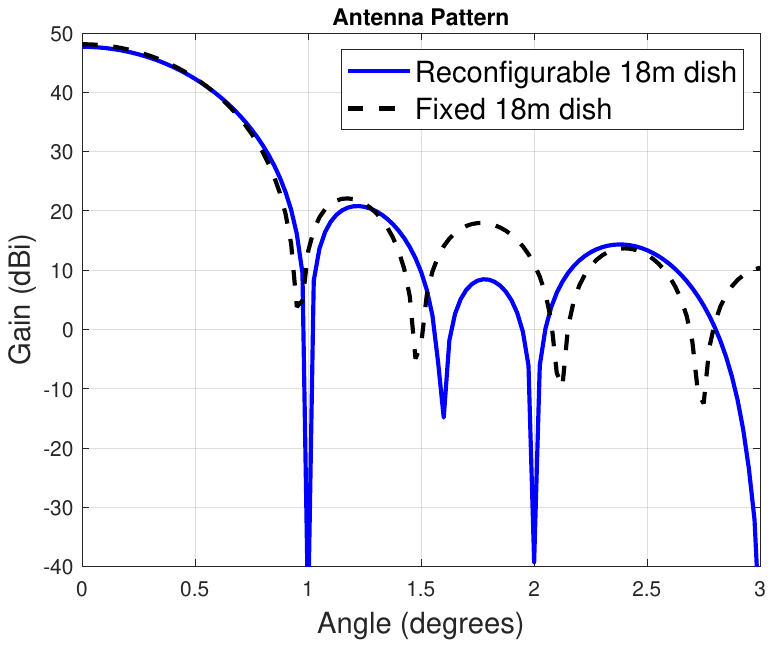}
    \caption{\rev{H-plane co-pol pattern for standard (fixed) 18m dish and reconfigurable 18m dish with 0.5m reconfigurable rim (continuously variable phase, main lobe control, $\psi=1.0^\circ, 2.0^\circ, 3.0^\circ$)}}
    \label{fig:fig4}
\end{figure}

As a last examination of the performance of the algorithms, we now examine the ability of the approach to mitigate main lobe gain variation seen in Figure \ref{fig:main_lobe}.  First, we use the gradient projection algorithm (as described above) with two constraints.  The first constraint is to maintain the main lobe gain at a constant value while the second constraint is to cancel the sidelobe at the angle $\psi$.  It was found that the main lobe gain is held constant at 47.8 dBi regardless of the angle of the null being created.  At the same time, to within machine precision, the gain at the null direction can be set to zero.

If quantized weights are used, the results are not quite as good as shown in Figure \ref{fig:fig6}.  While the main lobe gain is still maintained with strong consistency, the null depth is not as low as with  continuously-variable phase weights.  This is certainly expected based on the results of Figure \ref{fig:optimal_null}.  The quantized weights do not have the same capability to create deep nulls.  This is likely due to the fact that quantization limits the solution space and thus weights may not exist that can accomplish both constant main lobe gain and providing a deep null at a specific direction.  This is particularly true around the largest sidelobe at $1.25^o$ as was shown above.  As an example quaternary (i.e., $M=4$) weights were generated for main lobe control and a null at $\psi=1.75^\circ$.  The resulting the H-plane co-pol  pattern is plotted in Figure \ref{fig:fig7} along with the fixed reflector co-pol pattern.  As expected, a deep null is placed at $1.75^\circ$ while maintaining the desired main lobe gain.  

As an additional check on the impact of modifying the pattern, we also examined the and cross-pol pattern.  We can see in Figure \ref{fig:fig7} that while the weights do degrade the cross-pol pattern  (which is zero in the H-plane in the non-reconfigurable system) the cross-pol gain does not exceed $-22$ dBi. Further, it is possible that with proper constraints, this cross-pol degradation could also be controlled.


\begin{figure}
    \centering
    \vspace{1cm}
    \includegraphics[trim={5cm 8cm 5cm 10cm}, width=0.75\columnwidth]{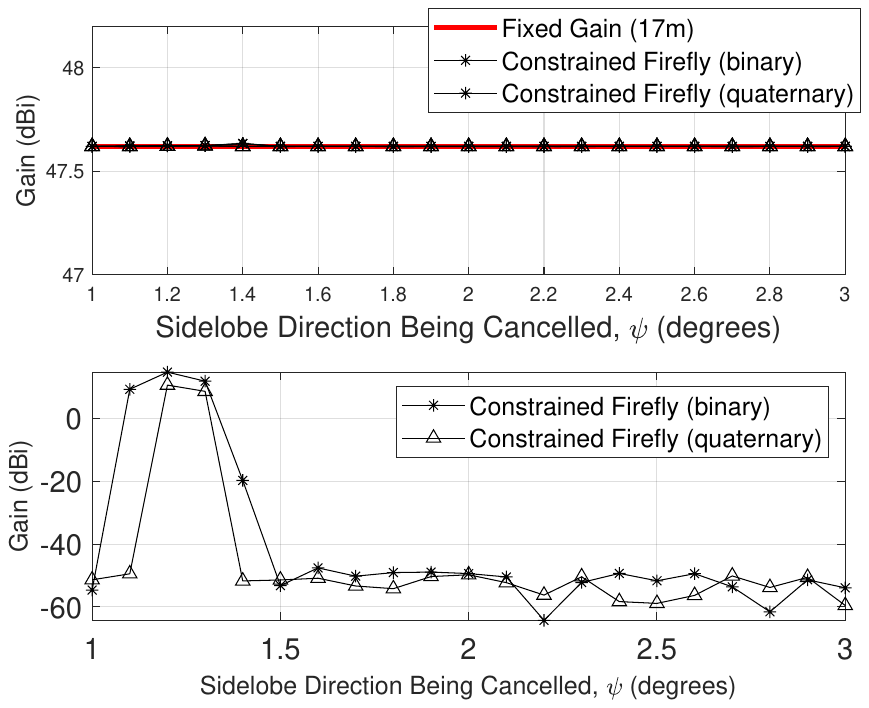}
    \caption{\rev{Main lobe gain $G(0)$ (top) and sidelobe gain (bottom) when cancelling sidelobe at $\psi$ while enforcing a main lobe constraint using Firefly and quantized weights ($M=2,4$) } }
    \label{fig:fig6}
\end{figure}

\begin{figure}
    \centering
    \vspace{1cm}
    \includegraphics[trim={5cm 8cm 5cm 10cm}, width=0.75\columnwidth]{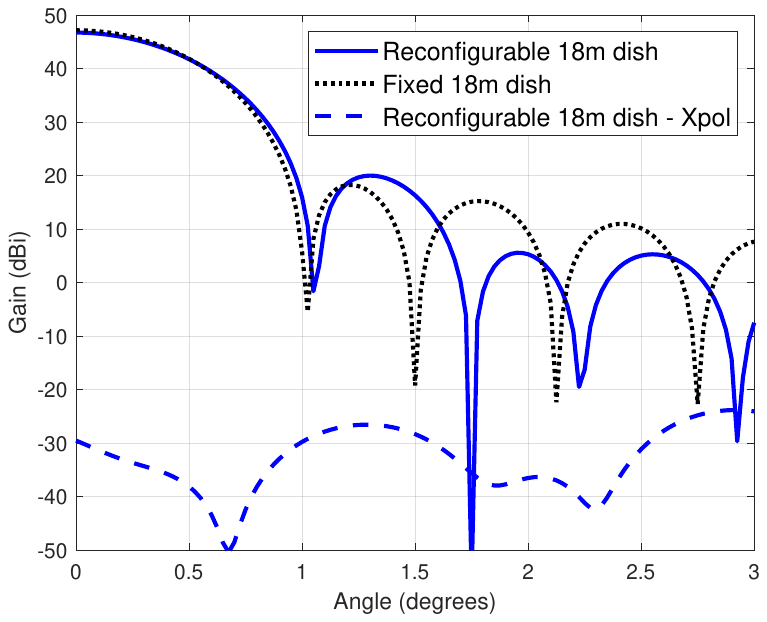}
    \caption{\rev{H-plane co-pol pattern for standard (fixed) 18m dish and H-plane co-pol and cross-pol patterns for reconfigurable 18m dish with 0.5m reconfigurable rim and quaternary Weights (main lobe control, $\psi=1.75^\circ$)}}
    \label{fig:fig7}
\end{figure}

\subsection{Frequency and Bandwidth Sensitivity}
In the examples above, the weights were chosen to optimize performance at a specific frequency.  In particular, the examples considered a frequency of $f=1.5$ GHz.  This raises the question as to whether or not the chosen weights are still effective at other frequencies. To explore the sensitivity to frequency, we used \rev{the Firefly Algorithm described above to find the optimal weights with $M=4$ for $\psi=1.75^\circ$.  We then examined the gain for $1.42{\text{ GHz}}\leq f \leq 1.58{\text{ GHz}}$.  The gain over this range of frequencies is shown in Figure \ref{fig:freq_sens} (see plot with 'x' markers).  From the plot we can see that although the null at the desired angle is reasonably deep (-50dBi) at the desired frequency (1.5 GHz), the null rises quickly as the frequency deviates the the design value.  For example, the null is only 1dB when the frequency deviates by 40MHz and 8dB at 80MHz from the design point. }

\begin{figure}
    \centering
    \vspace{1cm}
    \includegraphics[trim={5cm 8cm 5cm 10cm}, width=0.75\columnwidth]{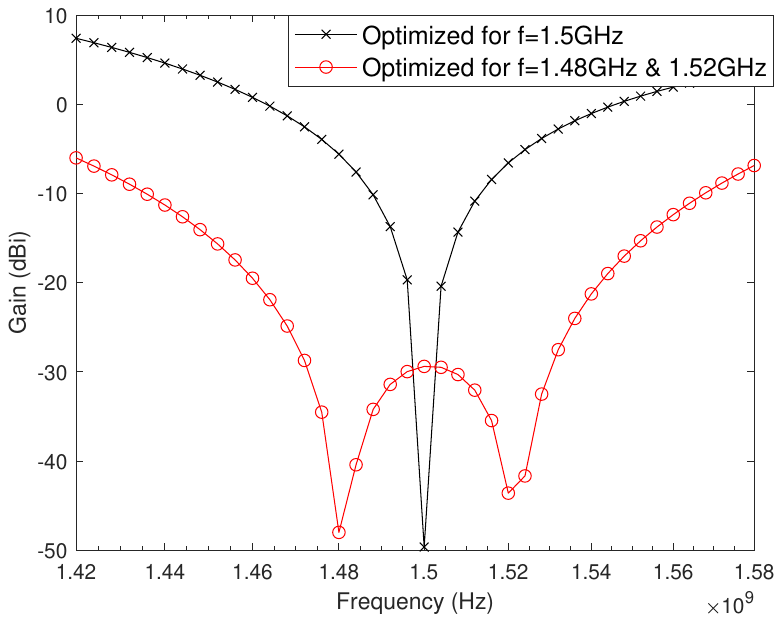}
    \caption{ \rev{Null depth for reconfigurable 18m dish with 0.5m reconfigurable rim and quaternary weights as a function of frequency ($\psi=1.75^\circ$)}}
    \label{fig:freq_sens}
\end{figure}

There are two ways that this could be improved, depending on the application needs.  If the application simply needs to change the center frequency of the null, but the signal of interest is narrowband, new weights could simply be generated at the new frequency.  In such a case, new weights could be found more quickly by starting from the weights at the current frequency.  The improvement in convergence time can be seen in Figure \ref{fig:freqconv}.  Specifically, in the figure three plots are shown.  The first plot (blue 'x') shows the gain achieved at each iteration of the weights when starting from a random set of weights and optimizing for $f=1.5$ GHz.  As seen, {approximately 290 movements (i.e., $M=290$)} are needed to reach a null depth of approximately $-60$ dBi.  Now, if we desire new weights optimized for $f=1.42$ GHz (80MHz away from the original frequency), we can start with the weights for $1.5$ GHz.  This shortens the convergence time from approximately \rev{290 movements to approximately 50 movements}.  Further, if the change in frequency is only 10MHz (i.e., $f=1.49$ GHz), the convergence time is similar.

\begin{figure}
    \centering
    \vspace{1cm}
    \includegraphics[trim={5cm 8cm 5cm 10cm}, width=0.75\columnwidth]{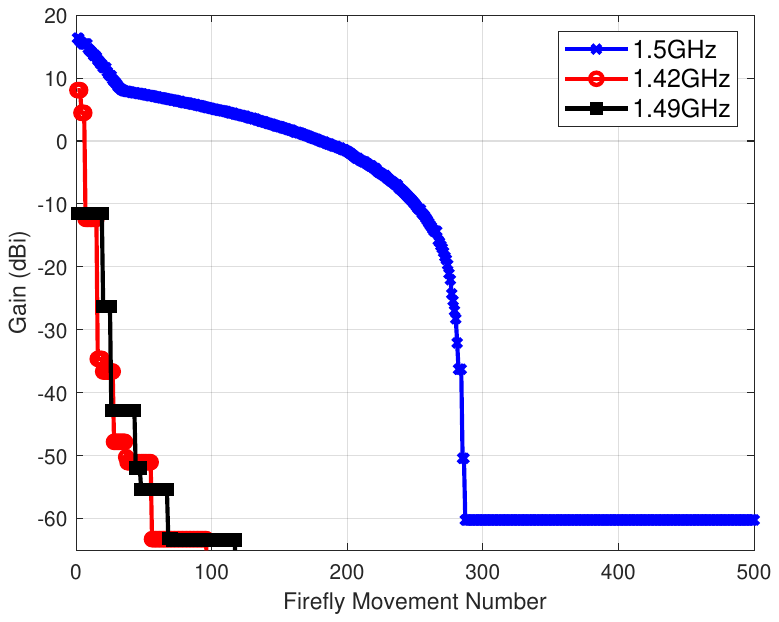}
    \caption{\rev{Convergence of quaternary Weights with simulated annealing when optimizing for different frequencies.  The optimization process for $f=1.5$ GHz starts from a random starting point while the optimization for $f=\{1.42, 1.49\}$ GHz start from the optimal weights for $f=1.5$ GHz.  ($\psi=1.75^\circ$)}}
    \label{fig:freqconv}
\end{figure}

If the application requires not simply a change in frequency, but instead a wider bandwidth, a second approach can be used.  Specifically, the optimization can be applied to multiple frequencies simultaneously.  An example of this approach is shown in Figure \ref{fig:freq_sens}.  In this example, the weights were optimized for $f=1.52$ GHz and $f=1.48$ GHz simultaneously.  As can be seen, the weights provide a null below $-30$ dBi over a bandwidth of approximately 40MHz.

\section{Closed-Loop Weight Selection} \label{sec:closedloop}
While we have shown that element weights for the reconfigurable portion of the dish can be found for placing nulls, the algorithm depends on the antenna pattern being well understood.  In practice  the antenna pattern is sensitive to perturbations and unmodeled effects and thus will be known with limited accuracy.  This will make  the field equations given above inaccurate. 

An alternative approach exists. Rather than evaluating the field equations for each considered set of weights, the weight selection algorithm can make decisions by measuring the interference power at the output while applying a specified set of weight vectors.  This ``closed-loop'' technique occurs in real-time and thus relies on accurate decision-making when evaluating the impact of different element weights  to mitigate interference power. Additionally, if the interference source is moving, the optimal weights will change over time and can be tracked by this approach.  In this section, we first determine the probability of error for the underlying decision process, and then provide simulation results obtained via a closed-loop optimization approach based on simulated annealing under several different scenarios.  Since convergence time is a major issue with the closed-loop approach, we then examine a second (faster) approach that uses a library of weight vectors which are measured during a calibration phase.

\subsection{Closed-Loop Simulated Annealing}
In order to address the closed-loop approach, we model the scenario as follows. 
We assume that the radio telescope employs some initial weight vector and is pointed at an astronomical source which has $S/N << 1$ in the bandwidth of interest and integration time (which is typically true).  Thus, the signal of interest can  be ignored for the purposes of nulling the interference. 
An interfering source is located outside the main lobe, but within the first few sidelobes. 
We also assume that a signal containing the interfering signal can be obtained from the back-end electronics. By integrating over a number of noise-corrupted received samples,  an estimate of the received interference power due to the current element weights  can be made. This interference power measurement is used as a decision metric in conjunction with a simulated annealing-based search as will be described. \rev{Note that population-based approaches are not desirable  in the closed-loop case since we must evaluate each set of weights in real time.  Since there is a high probability that some members of the population will not be good weights, each iteration of the population will have some values that will degrade performance.  Thus, we apply a version of simulated annealing.}

We can represent each of the $N$ received received samples $x[k]$ as 
\begin{equation}
    \label{eq:closed_loop_sig}
    x[k] = \sqrt{G_\psi(\bfw^{(t)})}z[k] + n[k]
\end{equation}
where $G_{\psi}( \bfw^{(t)})$ is the gain of the antenna in the direction of the interferer in the direction $\psi$ due to weights $\bfw^{(t)}$, $z[k]$ is the interference waveform and $n[k]$ denotes additive white Gaussian noise with variance $\sigma_n^2$. 

The software for the weight control algorithm must make measurements and estimate the power of the interfering source based on the current weights. 
However, the power estimate will be imprecise depending on the interference-to-noise ratio (INR). 
Because of this, the simulated annealing process will suffer from incorrect decisions. 
More specifically, simulated annealing involves comparing the previous cost $C(\bfw^{(t)})$ against a proposed cost $C(\bfw^{(tmp)})$. 
Previously, the cost was defined with respect to the gain provided in a specific direction due to a given set of  weights applied to the reconfigurable elements of the dish. 
Here, we propose the use of an output power  measurement for the cost.

\subsubsection{Decision Probabilities}
 
Let the interfering signal $z[k], k=1:N$ represent $N$ samples of a zero-mean Gaussian random variable with variance $\sigma_I^2$.  
The goal at each step of the simulated annealing process is to determine when a  set of proposed weights $\bfw^{(tmp)}$ provides \emph{lower gain} (and thus less interference power) than the current set of weights $\bfw^{(t)}$. 
To that end, let the cost metric for a set of weights at time $t$ be denoted as
\rev{
\begin{equation}
    \label{eq:closed_loop_cost}
    Z_t  = \frac{1}{N}\sum_{k=1}^N|x[k]|^2
\end{equation}
}
Let the gain in the interference direction $\psi$ at time index $t$ be $G_\psi(\bfw^{(t)})=G$, and the gain at time index $t=t+1$ be $G-\Delta G$, where $\Delta G$ represents the reduction in the gain due to changing the weights to the proposed weights. 

We can define the probability of error as the probability that the new set of weights is rejected despite the fact that the new gain in the interferer's direction is actually lower ({\it i.e.,} $\Delta G > 0$):
\begin{equation}
    \label{eq:prob_error}
    Pr(\textit{error}) = Pr(Z_1 > Z_0 | \Delta G > 0)
\end{equation}
Now, let the interference-to-noise ratio (INR) be given as ${\text {INR}} = \sigma_I^2/\sigma_n^2$.

\begin{theorem}
\label{thm:error_prob}
Given measurements $Z_0$ and $Z_1$ with $\Delta G > 0$, the probability of rejecting the new weight vector is
\begin{equation}
    \label{eq:p_error}
    Pr(Z_1 > Z_0 | \Delta G > 0) = Q\left(\Gamma \right)
\end{equation}
\begin{equation}
    \Gamma = \sqrt{\frac{N\Delta G^2\text{INR}}{\left(4G - 4G\Delta G + 2\Delta G^2\right)\text{INR} + 8G - 4\Delta G + \frac{4}{\text{INR}}}}
\end{equation}
where $Q(x)$ is the standard $Q$-function (i.e., the integration of the tail of a standard normal from $x$ to infinity).   
\end{theorem}
\begin{proof}
    See Appendix.  
\end{proof}


Figure \ref{fig:error_prob_delta} plots the probability of error for three different {\text{INR}} values, and fixed $\Delta G=0.01$ and $G_\theta(w[k]) = G = 1$ as $N$ increases. 
For low values of $N$, $P_e$ remains near 0.5, but as $N$ increases past $10^4$ to $10^5$, the probability of error improves significantly. 
This indicates that longer integration times can substitute for higher INR. Simulation has shown that this expression is very accurate.

Further, Figure \ref{fig:gain_convergence} demonstrates that the closed-loop simulated annealing algorithm indeed benefits from the improved probability of error provided by longer integration times. 
However, even with large integration times,  the $P_e$ remains high for for many iterations, which results in longer search times.   This is because, as seen from Eq. (\ref{eq:p_error}),  $G$ and $\Delta G$ have a large effect on the probability of error.  While the shape of the curves in Figure \ref{fig:gain_convergence} is similar for $G<1$, the $G\geq 1$ regime is much more error-prone due to the higher gain.  Once $G$ has been reduced to a sufficiently small value (typically below 1), the probability of error decreases rapidly and the weights  converge to a gain similar to the open-loop technique.

\begin{figure}
    \centering
    \includegraphics[scale=0.5]{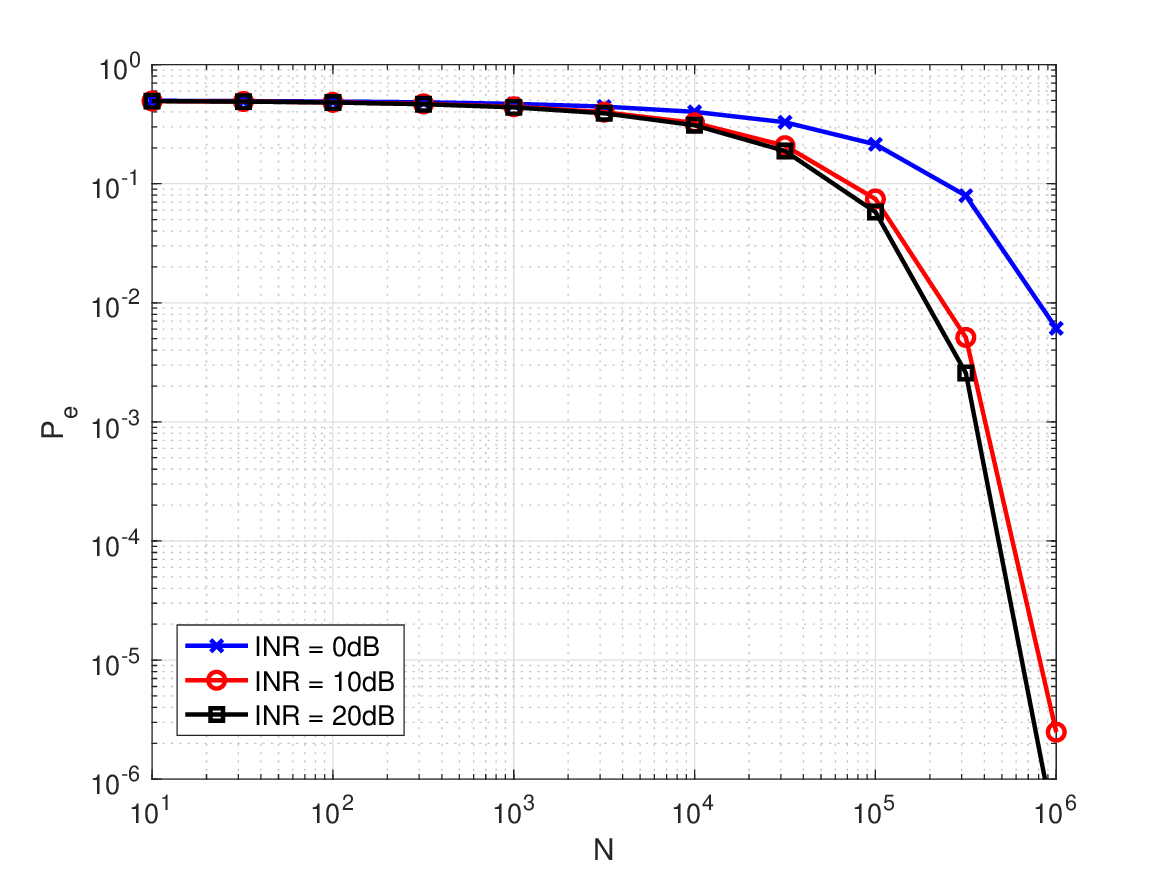}
    \caption{\rev{The probability of error decreases as the sample length $N$ increases. The single-sample INR ranges from 0dB to 20dB, $\Delta G $ = 0.01 and $G_\theta(w[k])=0$dBi}}
    \label{fig:error_prob_delta}
\end{figure}

\begin{figure}
    \centering
    \vspace{1cm}
    \includegraphics[trim={5cm 8cm 5cm 10cm}, width=0.9\columnwidth]{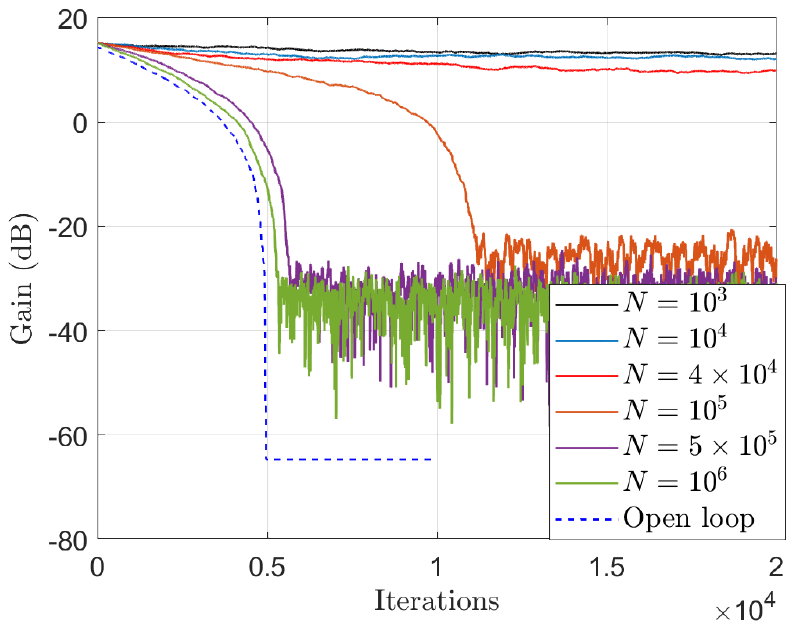}
    \caption{\rev{Increasing the number of samples per decision $N$ reduces both the number of iterations and the final gain. Sample length ranges from $10^3$ to $10^6$. INR = 10$dB$, $\theta=1.75^\circ$. }}
    \label{fig:gain_convergence}
\end{figure}

\begin{table}[h]
    \centering
    \begin{tabular}{l|l}
    \hline 
    \multicolumn{2}{c}{Closed-Loop Simulated Annealing} \\
    \hline
     1. & Initialization: $k=0$, $\bfT=\left [ 1, \frac{1}{2},\frac{1}{3},\ldots, \frac{1}{T} \right ]$ \\
      & \hspace{1.5cm} $\bfw^{(0)}$ randomly chosen from ${\cal W}^N$ \\
      & \hspace{1.5cm} $E^{best} = Z_0=\frac{1}{N}\sum_{n=1}^N|x[n]|^2$ \\
         & Repeat \\
      2. & $\hspace{1cm} \bfw^{tmp} = \bfw^{(k)}$\\
      3. & $\hspace{1cm} n=\lceil rand*N \rceil$,$m=\lceil rand*M \rceil$ \\
      4. & $\hspace{1cm} \bfw^{tmp}_n={\cal W}(m)$ \\
      5. & $\hspace{1cm} E^{tmp} = Z_{k+1}=\frac{1}{N}\sum_{n=1}^N|x[(k+1)N+n]|^2$\\
      6. & $\hspace{1cm} \Delta E =E^{best} - E^{tmp} $ \\
      7. & $\hspace{1cm} \text{if } \Delta E \geq 0$ OR $\text{rand}(0,1) < \min \left(1, e^{\Delta E /\bfT(k)}\right )$ \\
      8. & $\hspace{1.5cm} \bfw^{(k+1)}=\bfw^{tmp}$ \\
      9. & $\hspace{1.5cm} E^{best} = E^{tmp}$ \\
      10. & $\hspace{1cm} \text{else } \bfw^{(k+1)} = \bfw^{(k)}$ \\
      11. & $\hspace{1.5cm} k=k+1$ \\
          & Until $k=T$  \\
      \hline
    \end{tabular}
    \label{tab:closed_loop}
\end{table}

\subsubsection{Single-Angle Results}
Given a sufficient number of iterations, the proposed approach (closed-loop simulated annealing) is able to obtain performance approaching that of the alternative ``open-loop'' annealing, provided the $\text{INR}$ is sufficiently high.  
As shown in Fig. \ref{fig:theta_INR}, an $\text{INR}$ of 20 dB is sufficient to obtain null depts of roughly $-50$dB with 16-ary weights (except at the first sidelobe). A lower or higher INR (10dB or 30dB) does not make a substantial difference.
The interferer is assumed to be at a constant angle with respect to the receiver over many ($10^5$) iterations.  The gain plotted is the gain after those $10^5$ iterations.  Note that the total number of samples needed for $N_I$ iterations depends to some degree on the INR per sample.  The simulations assume that the INR per decision is that listed in the legend.
Since the first sidelobe has high gain, we believe that the number of solutions is smaller and thus a larger number of error-free iterations is needed to find the optimal solution.  Thus, the resulting gain after $10^5$ iterations is higher. 
The fact that the open-loop annealing process can find a very low gain vector supports this belief; the solution does exist, it is just more difficult to find using closed-loop annealing. 

\begin{figure}
    \centering
    \includegraphics[scale=0.65]{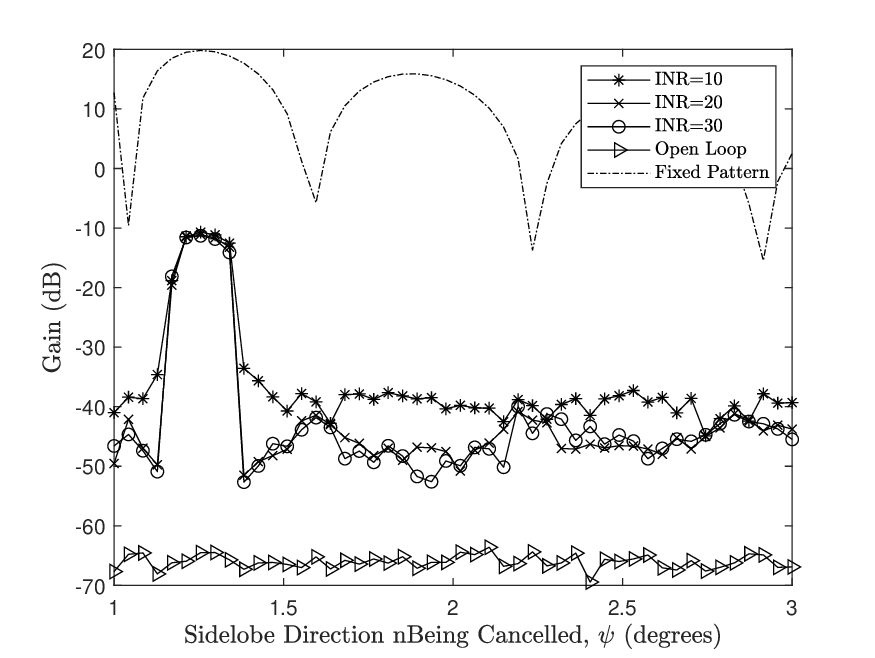}
    \caption{\rev{Gain obtained in the interference direction for Varying Values of \texttt{INR} with 16-ary Weights. We can see that 30dB is sufficient to essentially match open-loop performance. }}
    \label{fig:theta_INR}
\end{figure}

\subsubsection{Clustering to Improve Convergence Time}
The simulated annealing approach (a) initializes each weight with a random value, (b) proceeds by assigning one weight in the vector with a random value and the (c) determines whether or not to accept the new vector of weights based on the cost function and the cooling temperature.  For large sets of weights, convergence can be slow as shown in Figure \ref{fig:gain_convergence}.   
To improve convergence time, we explore \emph{clustering} rim surfaces into fixed groups which use the same value in order to reduce the search space and corresponding convergence time. 
In this approach, each annealing iteration  selects and modifies an entire cluster simultaneously to the same value. 
Figure \ref{fig:clustering} demonstrates the effectiveness of this technique, for a specific angle ($\theta = 2.0^\circ$).  
The plot presents the gain (i.e., null-depth) achieved after $10^4$ iterations for different weight quantization values ($M$) using various cluster sizes and an INR of 30dB.      
We can see that for cluster sizes up to approximately 50, there is no loss in the gain achievable, except in the binary weight case.  For binary weights, a cluster size beyond approximately 35 shows a loss in achievable null depth.

The major effect of this clustering is to reduce the problem space by several orders of magnitude.  The reduction in the problem space improves convergence time.  In fact, the results in Figure  \ref{fig:clustering} used only $10^4$ iterations as opposed  to $10^5$ in Figure \ref{fig:theta_INR}.   
The reason that clustering improves convergence without impacting null depth appreciably is that while reducing the problem space  reduces the degrees of freedom,  due to the geometry of the problem, nulling specific directions does not seem to require as many degrees of freedom as are available.  This conjecture is supported by the results of Figure \ref{fig:clustering} and implies correlation  between the needed weights by adjacent surfaces on the reconfigurable rim due to the geometry of the antenna structure. 

\begin{figure}
    \centering
    \includegraphics[scale=0.6]{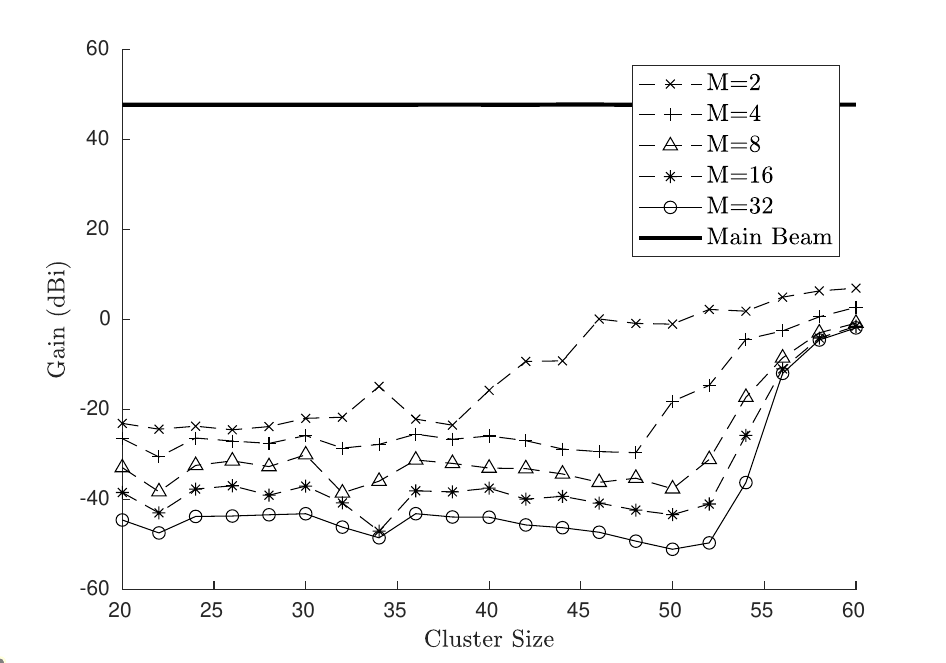}
    \caption{Null-depth Achieved after $10^4$ Iterations for Different Weight Quantization Values ($M$) using Different Cluster Sizes.  ($\theta = 2.0^\circ$, INR = 30dB)}
    \label{fig:clustering}
\end{figure}

The over-dimensionality of the single-angle problem is also supported by Figure \ref{fig:average_weights}.  This figure was created by running non-clustered open-loop simulated annealing with a source at $2^\circ$, $10^3$ times with binary weights.  The resulting set of $10^3$ weight vectors was then averaged.  Since the simulated annealing process is stochastic and starts at a random point, there is no guarantee that the final weight vectors of multiple runs will be identical or even similar.  
In fact, in this set, no solution occurs twice and the average normalized Hamming distance between solutions is 0.09.   However, Figure \ref{fig:average_weights} clearly shows that the individual elements of the solutions do exhibit structure.  If the individual elements of weight solution vectors were equally likely to be $\{-1,1\}$, we would expect Figure \ref{fig:average_weights} to show values nearly equal to zero.  However, we can see that the values are nearly +1 or -1, meaning that while there may be many optimal solutions, the set of ``ideal'' weights have individual elements that are heavily biased towards certain values.   
Additionally, this figure shows the final values for each weight tend towards clusters, even when clustering is not forced. 

Interestingly, while it seems that the weights tend to converge towards this structure, no optimal set of weights is actually equal to the set of most likely values.  In fact, the pattern gain provided by using the set of most common weight values of each surface is $\sim 15$ dBi, which is no better than using randomly chosen weights. 
However, when the resulting collection of final weights are treated as a distribution (i.e., individual weights are chosen from a binary distribution indicated by Figure \ref{fig:average_weights}), the expected pattern gain drops to approximately $-20$ dBi. More specifically, Figure \ref{fig:gain_dist} presents gain histograms for weights which are randomly generated (i.e., the intial weights), weights acquired via closed-loop simulated annealing, and finally weights which determined by sampling from a binary distribution implied by the values shown in Fig. \ref{fig:average_weights}. 
Interestingly, using this latter approach yields gains which are significantly lower than randomly selected values (as expected), but not as low as those generated through simulated annealing.


These clustering techniques provide a way to reduce the problem complexity while maintaining good performance. 
However, even using this approach, the convergence times are still relatively long and require a non-moving source. 
In realistic implementations, sources will likely be moving with respect to the receiver. 
We consider this situation in the following section.

\begin{figure}
    \centering
    \includegraphics[scale=0.55]{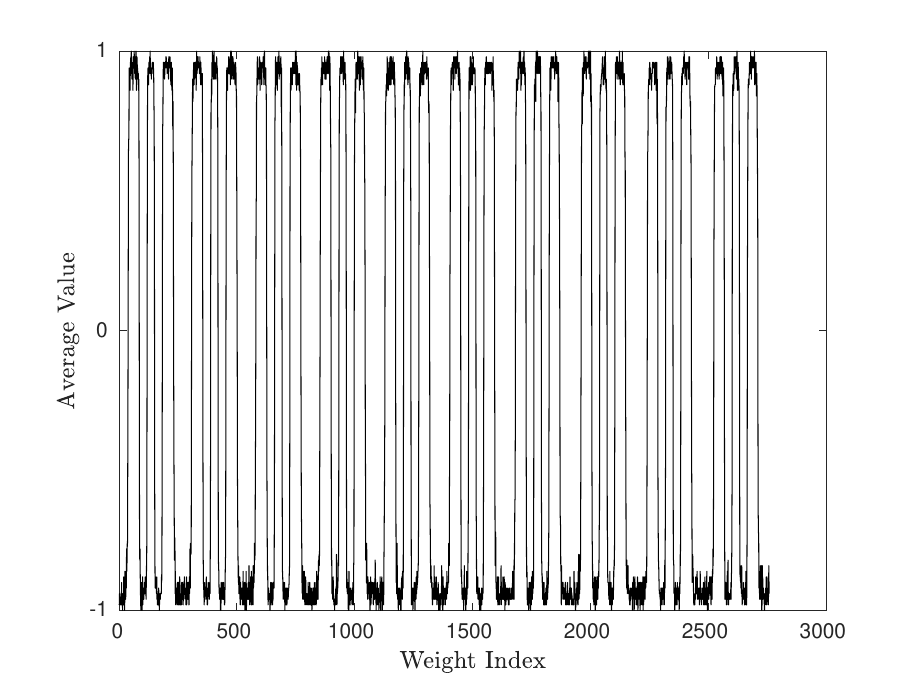}
    \caption{Average Final Weights for $10^3$ Independent Runs of Closed-Loop Simulated Annealing with Binary Weights when Placing a Null at 2$^\circ$. }
    \label{fig:average_weights}
\end{figure}
\begin{figure}
    \centering
    \includegraphics[scale=0.55]{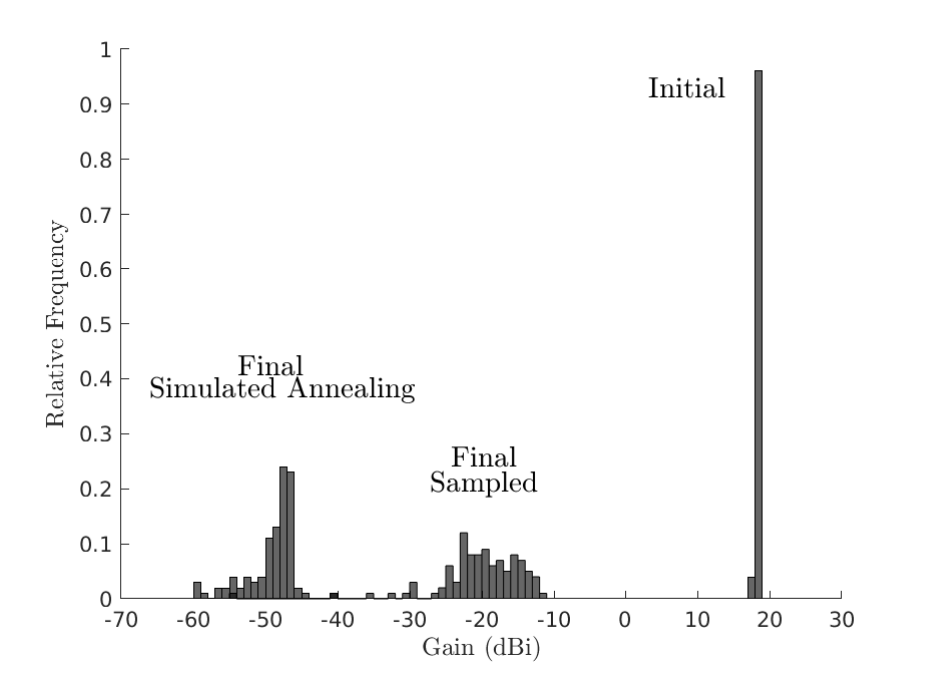}
    \caption{Histograms for the gain distribution of closed-loop simulated annealing, random weights, and samples from the most common final weights. }
    \label{fig:gain_dist}
\end{figure}

\subsubsection{Performance with Moving Source}
Closed-loop annealing requires many iterations for convergence and thus requires many samples of the interfering signal at a stationary (or near stationary) angle of arrival.   However, the interferers of interest are assumed to be in Low Earth Orbit (LEO), and will thus move across the field of view of the radio telescope. 
As has been shown in this work, the antenna pattern for a fixed set of rim weights can have a null at arbitrary locations (outside the main beam) and even at multiple angles, but not necessarily at all angles simultaneously.  Thus, the antenna system will need to update the vector of weights at regular intervals (depending on the source angular velocity) to maintain the pattern null at the correct location as the source moves across the sky. 

To examine the performance of closed-loop simulated annealing in the moving source case, we simulated the algorithm as an interferer moves across the view of the receiver.  If the algorithm can fast enough to keep up with the movement of the source,  the energy received from the interferer will remain low.
In practice, we find that this method requires a very high number of decision intervals. 
In order to allow the required number of decisions, the interference likely needs to be oversampled with weights changing every few samples.

Fig. \ref{fig:moving_int} shows the performance of closed-loop simulated annealing with a time-varying interference angle. 
The weights are initialized with a low-gain weight vector for the initial angle of $3^\circ$.  
We assume that the sampling rate is 1.1GHz and that the source moves with an angular velocity of 0.79 degrees/sec and the INR is 60dB.  
Thus, a decision is made every $N=10^3$ samples and there are $2.2\times10^6$ decisions per second. 
The large number of decisions per second is required to maintain low gain, since if the interferer moves too far between decisions the null in the pattern will be lost. 
This also necessitates the high sampling rate, to allow enough samples per decision. 
We can see that when 8-ary weights are used,  this technique is able to maintain a null depth below -35dBi (except when interferer is in the first sidelobe) while the depth is pushed down below -50dBi for 16-ary weights. 


%
\begin{figure}
    \centering
    \includegraphics[scale=0.55]{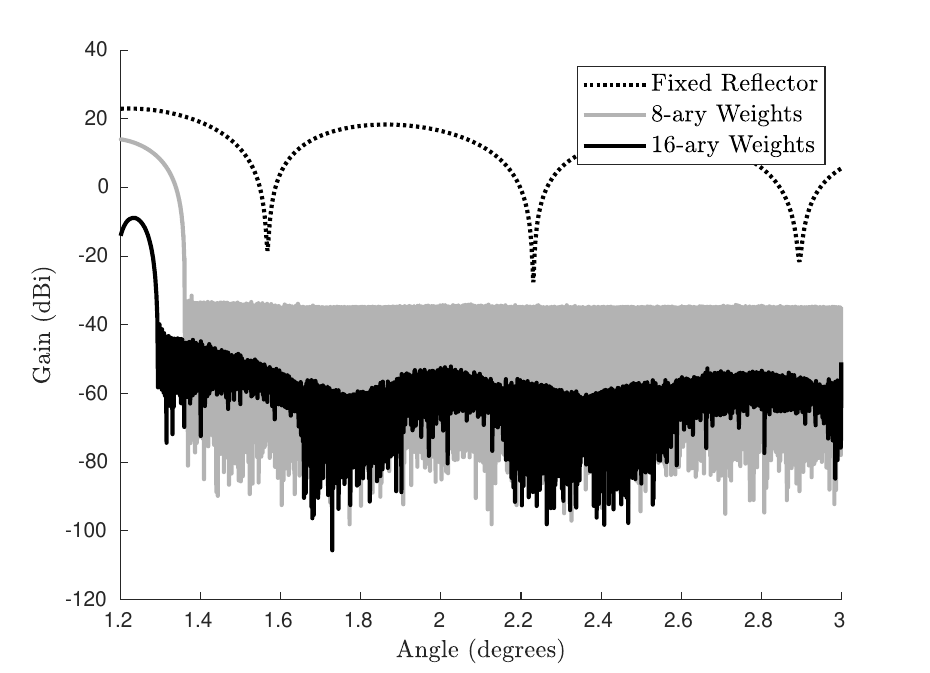}
    \caption{Gain in the direction of an interferer as the interferer moves across the sky (from $\theta=3^\circ$ to $\theta=1.2^\circ$, using closed-loop simulated annealing. The INR is 60$dB$ and $N=10^3$. }
    \label{fig:moving_int}
\end{figure}

\subsection{Library-based Approach}
The previous approach was successful, but relied heavily on oversampling and a relatively high INR.  In fact, the sampling rate was 110 times the bandwidth of the interfering signal.  The high oversampling rate and high INR were needed for the algorithm to converge sufficiently fast (i.e., fast enough to keep up with the moving interferer.)     

A second approach that allows faster convergence relies on the ability to create a library of weight vectors, each calibrated to provide a null in a desired direction. 
In practice, this calibration could be accomplished using a stationary source.  This method also uses a simulated annealing decision process. 
However, instead of pseudo-randomly creating the new weight each decision time, this method samples a new vector from the pre-calibrated library. 
Since the library can be ordered by null location, this sampling can be structured, with a higher chance of sampling vectors with nulls close to the current vector.

Fig. \ref{fig:library} shows the performance of this technique for the same moving source described above. 
Specifically: the interferer bandwidth is 10MHz, the INR per decision is 30 dB, and the interferer moves at 0.79 deg/sec.   
This yields $\sim$2$\times10^4$ decision intervals per second, as compared to $\sim$2.2$\times 10^6$ decision intervals per second in the closed-loop simulated annealing approach described earlier. 
Since adjacent elements in the library have a larger difference in gain than changing a single weight would provide, a smaller number of decisions per second can provide similar results. 
Thus, the sampling rate can be reduced by two orders of magnitude.
On the other hand, the performance of this method is inferior compared to the energy minimization technique described above due to large peaks in gain. 
These are unavoidable as exploration of the library forces the system to accept the gain of the explored weights which may be high. 
This is also a consequence of the larger change in gain between library elements as opposed to the previous approach where only a single element in the weight vector was changed. 
However, this performance is obtained under reasonable INR conditions and with a substantially reduced number of decision intervals. 

\begin{figure}
    \centering
    \includegraphics[scale=0.55]{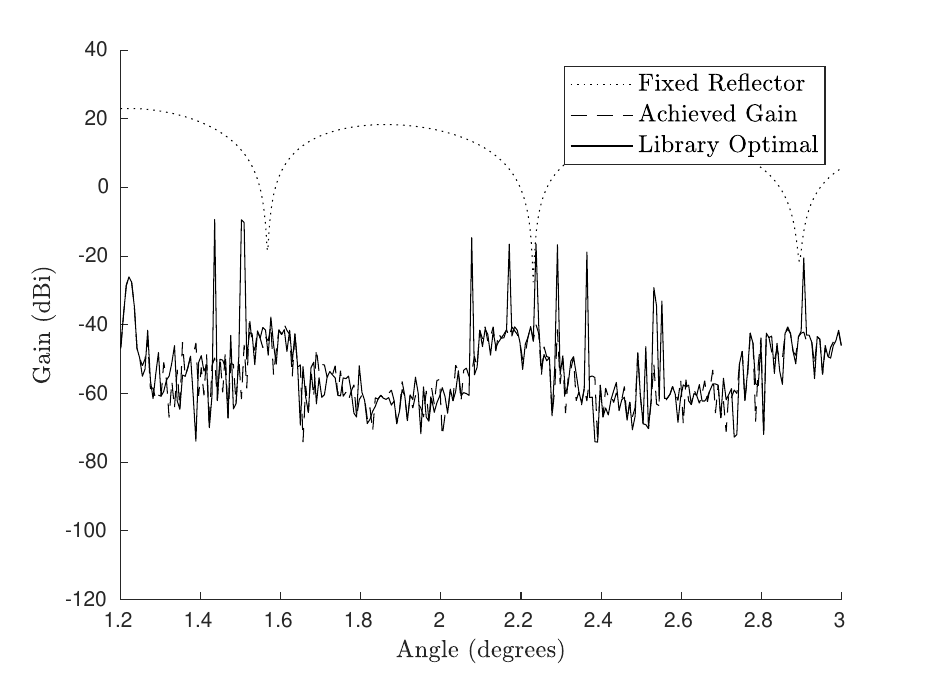}
    \caption{Gain in the direction of an interferer moving across the sky, using the library-based approach. The INR is 30$dB$ and $N=10^3$. }
    \label{fig:library}
\end{figure}

\section{Hybrid Approach to Weight Selection} \label{sec:hybrid}
In the previous sections we have demonstrated that open-loop approaches can be used to determine optimal weights for the rim scattering, provided that the electric field intensity (i.e., the antenna pattern) is known.  Alternatively, closed-loop approaches can be used to determine the weights provided that the output interference power can be measured as the weights are adapted.  This approach obviates the need for knowledge of the pattern, but requires the signal to be sampled at a rate much faster than the change in optimal weights due to movement of the interfering source.  

These two approaches make an ``all or nothing'' assumption on the knowledge of the pattern.  However, while exact knowledge of the pattern is overly optimistic, assuming no knowledge of the pattern is perhaps overly pessimistic.  A hybrid approach would be to use the open-loop approach to provide a starting point for the closed-loop technique.  The goal would be to use knowledge of the approximate field to provide a good starting point for the closed-loop annealing process.

To test this approach we used the following model.  The open-loop optimization used a presumed pattern to determine optimal weights.  Specifically, the presumed pattern is calculated assuming that the true surface current on the dish is related (but not equal) to the assumed surface current distribution.  The assumed pattern was calculated using a surface current resulting from a feed with taper $q=1.5$, while the true pattern results from a feed with taper $q=1.14$.  Both patterns are calculated using a PO approach, but $q=1.14$ corresponds to the true pattern while $q=1.5$ corresponds to the assumed pattern (i.e., what our best understanding of the antenna predicts).  Note that the patterns with two different feed tapers are simply being used as proxies for the true and assumed patterns.  We are not arguing that knowledge of the feed taper is what makes the pattern difficult to know in practice. 

In Figure \ref{fig:ideal} we plot four patterns.  Specifically, we first plot the ``assumed'' and  ``true'' quiescent patterns (i.e., the fixed patterns without optimization)   assuming the two feed tapers listed above.  We can see that the assumed pattern has slightly lower gain than the true pattern. Additionally, we plot the pattern of the reconfigurable dish where the reconfigurable dish weights were optimized to create a null in the assumed pattern at $\psi=1.75^\circ$ using simulated annealing with $M=4$.  In other words, the weights were calculated assuming $q=1.5$. This plot is labeled ``Reconfigurable 18m dish - Assumed Pattern'' and represents what the optimization algorithm believes the pattern to be after optimization.  The final plot included is labeled ``Reconfigurable 18m dish - True Pattern'' and represents the actual pattern achieved when applying the calculated weights from the optimization process. The mismatch between the assumed and true patterns results in the null depth being substantially reduced.  While the algorithm believes that a null below -50dBi in gain was created, in actuality the gain is approximately -9.3dBi.     Thus, although the ``optimal'' weights applied to the assumed pattern results in a deep null at the desired angle, applying those same weights to the true pattern does not provide a substantial null.  However, while the null depth is obviously not what is desired, the overall true pattern is not dramatically different from the assumed pattern, which motivates a hybrid approach where inexact knowledge can be used to get the weights ``close'' to optimal and then closed-loop adaptation tweaks the weights to generate the desired null.  

\begin{figure}
    \centering
    \vspace{0.25in}
    \includegraphics[trim={5cm 8cm 5cm 10cm}, width=0.75\columnwidth]{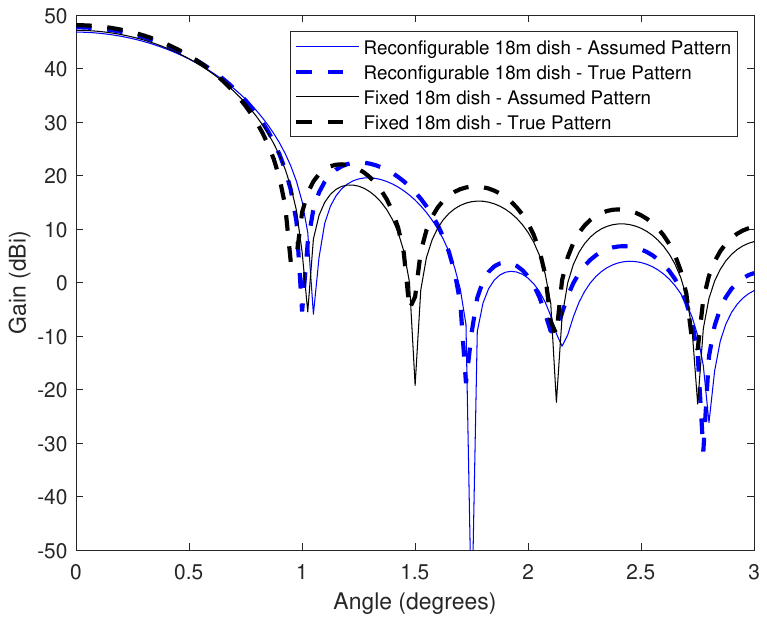}
    \caption{Impact of inexact knowledge on optimized pattern.  ``Assumed pattern'' refers to the pattern assumed by the weight optimization algorithm with (``Reconfigurable'') and without (``Fixed'') weight optimization. ``True pattern'' refers to the pattern that is actually achieved when using either using (``Reconfigurable'') or not using (``Fixed'') the weights obtained by the optimization algorithm. 
 (In both cases the reconfigurable weights optimized to create a null at  $\psi = 1.75^\circ$ using the assumed pattern).  }
    \label{fig:ideal}
\end{figure}

Thus, in the proposed hybrid approach,  weights are first obtained by using simulated annealing to create a null at the desired angle in the assumed pattern.  These weights are then used as a starting point for the closed-loop algorithm.  To demonstrate the effectiveness of this approach, Figure \label{fig:conv} plots the convergence of the weights for the two portions of the algorithm.  Specifically, the figure includes two plots.  The first plot is the convergence of the open-loop portion that finds the weights to create a null using the assumed pattern.  Note that the gain of the assumed pattern at $\psi = 1.75^\circ$ is plotted.  The convergence requires slightly over 6000 iterations due to the fact that it starts from a random initial set of weights.    As in the plot of Figure \ref{fig:ideal}, the null depth that the open-loop algorithm believes it achieves is well under -50dBi.  However, when applying those weights to the true pattern the actual null depth achieved is roughly -10dBi. However, these weights are not used, but instead serve as the starting point for the second portion of the hybrid algorithm. 

The second part of the algorithm starts by applying the open-loop solution, but estimates the interference power and modifies the weights using a closed-loop approach.  Note that in this example, we have assumed perfect gain measurements in order to isolate the improvement in convergence due to the starting point.  In the true closed-loop implementation, decision errors will lengthen the convergence.  However, the number of error-free steps required reduces from over 6000 when staring from a random set of weights to just under 600 when starting from the open-loop solution (i.e., an order of magnitude improvement).  This convergence time is also dramatically faster than the closed-loop approach described in the previous section.  The improvement is directly due to the fact that the open-loop optimization provides a much better initial starting point than we would have without it.  Thus, it appears that a hybrid approach may be a viable alternative to either open or closed-loop approaches. 

\begin{figure}
    \centering
    \vspace{0.25in}
    \includegraphics[trim={5cm 8cm 5cm 10cm}, width=0.75\columnwidth]{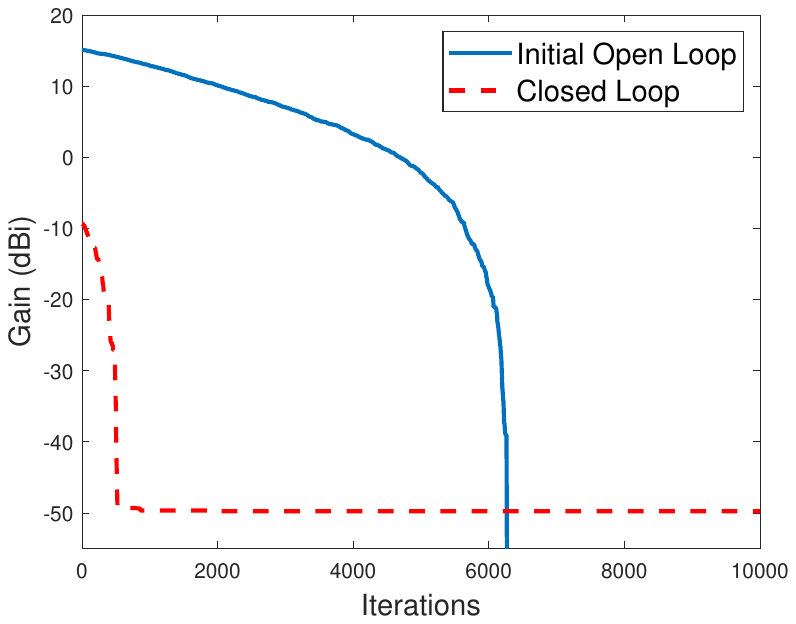}
    \caption{Convergence of the weights in the open-loop and closed-loop portions of the hybrid algorithm (open-loop weights optimized to create a null at  $\psi = 1.75^\circ$ using ideal pattern).  }
    \label{fig:conv}
\end{figure}

\section{Conclusions}
\label{sec:concl}

In this paper we have described multiple techniques for determining optimal or near-optimal weights for creating nulls in the pattern of a prime focus-fed circular axisymmetric paraboloidal reﬂector antenna equipped with reconfigurable elements on the rim.  The general approach applies to a larger class of reflector antennas (e.g., different feed systems), but the discussion was limited to this specific type for demonstration purposes.  It was shown that if the elements placed on the rim of the reflector are capable of imparting both amplitude and phase adjustment, creating perfect nulls at arbitrary directions of arrival is possible.  Multiple nulls can be created while the gain of the main lobe is simultaneously held fixed.  

More importantly it was shown that more practical unit-modulus weights can be found using a least-squares approach based on the gradient projection algorithm.  These weights, while not requiring gain control (i.e., they require phase shifts only), and capable of placing a null at any angle, unfortunately require continuously-variable phase.  However, it was also shown that deep nulls can be created using quantized unit modulus weights which require only binary or quaternary phase values using simulated annealing.  Further, it was shown that the approach can also be used to simultaneously create a null at a specific angle while maintaining a constant main lobe gain as well as creating nulls at multiple angles.  The frequency sensitivity of the weights was also examined along with techniques to address that sensitivity.    

The first group of weight selection approaches discussed in Section \ref{sec:weights} require knowledge of the antenna pattern and  calculated weights in an ``open-loop'' fashion.  To address  pattern uncertainty, in Section \ref{sec:closedloop} we  investigated ``closed-loop'' optimization, where no knowledge of the antenna pattern is required. 
Through this method, a receiver iteratively attempts to reduce the total power received from interference. 
Since this takes many iterations while the interferer  is constantly in motion, we  validated the tracking ability of the closed-loop approach through dynamic optimization, changing the angle of arrival of the interferer each iteration and showing that closed-loop annealing is able to maintain a consistent null.   

In order to improve the convergence speed of the closed-loop approach we also explored a library-based approach that chooses  weights from a fixed library of weights.  This approach was shown to also be effective and converged more quickly.  The faster convergence requires a lower sampling rate relative to the source movement.   

As a final investigation we explored a hybrid approach that uses an open-loop optimization to obtain an initial starting point for the weights in the closed-loop optimization.  It was found that this approach allows for a reduction in the convergence time that is multiple orders of magnitude and is a promising means for actual implementation.  

\rev{It should be noted that the current work is based on physical optics modeling of a reflector antenna.  In two independent evaluations,  we have shown that full-wave-based results are consistent with the above-mentioned physical optics-based results \cite{Hum23,Budhu24,Budhu24b}.  Specifically, it was shown that the simple serial search approach provides nearly the same performance in either case.  Based on those results, we have no reason to believe that the optimization techniques described in this work might give significantly better or worse outcomes when applied to a system modeled using full-wave simulation as opposed to PO. It is our opinion that the key is showing that the optimization algorithms can find a good solution, since the performance of the algorithm will not be fundamentally different regardless of the finer details of the electromagnetics.  Of course, prior to any implementation of the proposed optimization algorithms, adopters of this technology are encouraged to consider full-wave analysis as a final step for additional confirmation.}

\section*{Acknowledgements}
This material is based upon work supported in part by the National Science Foundation under Grant AST 2128506.  The authors would also like to thank Xinrui Li for his help with proof-reading and creation of the figures. 
\bibliographystyle{IEEEtran}
\bibliography{refs}

\onecolumn
\appendix
\section{\\Proof of Theorem \ref{thm:error_prob}}
\label{ap:thm}
\begin{theorem*}
Given measurements $Z_0$ and $Z_1$ with $\Delta G > 0$, the probability of rejecting the new weight vector is
\begin{equation}
    Pr(Z_1 > Z_0 | \Delta G > 0) = Q\left(\sqrt{\frac{N\Delta G^2\text{INR}}{\left(4G - 4G\Delta G + 2\Delta G^2\right)\text{INR} + 8G - 4\Delta G + \frac{4}{\text{INR}}}}\right)
\end{equation}
\end{theorem*}
\begin{proof}
Let the received signal be represented by $x = \sqrt{G}z[k] + n[k]$, where $k=1:N$, $z[k]\sim N(0, \sigma_I^2)$, $n\sim N(0, \sigma_n^2$) and $G$ is the gain. 
The objective is to determine whether the gain in a current measurement is lower than the gain in a previous measurement. 
To that end, let the previous gain be denoted as $G$, and the current gain denoted as $G-\Delta G$. 
Then, represent the current decision variable as $Z_1$ and the previous decision variable as $Z_0$. 
\begin{align}
    Z_0 &= \frac{1}{N}\sum_{k=1}^N \left(\sqrt{G}z[k] + n[k]\right)^2\\
    Z_1 &= \frac{1}{N}\sum_{k=1}^N \left(\sqrt{G-\Delta G}z[k] + n[k]\right)^2
\end{align}
Note that $Z_0, Z_1$ are chi-square distributed. 
So, as $N\to\infty$, $Z_0$ and $Z_1$ can be approximated as Gaussian random variables with means 
\begin{align*}
    \mu_{Z_0} &= G\sigma_I^2 + \sigma_n^2\\
    \mu_{Z_1} &= (G-\Delta G)\sigma_I^2 + \sigma_n^2
\end{align*}
and variances 
\begin{align*}
    \sigma_{Z_0}^2 &= \frac2N (G\sigma_I^2+ \sigma_n^2)^2\\
    \sigma_{Z_1}^2 &= \frac2N ((G-\Delta G)\sigma_I^2 + \sigma_n^2)^2
\end{align*}
When $\Delta G > 0$, we would like $Z_0 < Z_1$. 
Represent the probability of error as 
\begin{align}
    P_e &= Pr(Z_1 > Z_0 | \Delta G > 0)
\end{align}
Note that $Z_1 - Z_0$ is normally distributed as $N\to\infty$ with mean $\mu_{Z_1}-\mu_{Z_0}$ and variance $\sigma_{Z_1}^2+\sigma_{Z_0}^2$. 
\begin{align*}
    \mu_{Z_1}-\mu_{Z_0} &= (G-\Delta G)\sigma_I^2 + \sigma_n^2 - (G\sigma_I^2 + \sigma_n^2)\\
    &= -\Delta G\sigma_I^2
\end{align*}
\begin{align*}
    \sigma_{Z_1}^2 + \sigma_{Z_0}^2 &= \frac2N \left[\left((G-\Delta G)\sigma_I^2 + \sigma_n^2\right)^2 + \left(G\sigma_I^2 + \sigma_n^2\right)^2\right]\\
    &= \frac2N \left[\left(2G^2 - 2G\Delta G + \Delta G^2\right)\sigma_I^4 + \left(4G-2\Delta G\right)\sigma_I^2 \sigma_n^2 + 2\sigma_n^4\right]
\end{align*}
Then, 
\begin{align*}
    Pr(Z_1 > Z_0) &= \int_0^\infty \frac{1}{\sqrt{2\pi(\sigma_{Z_1}^2 + \sigma_{Z_0}^2)}} \exp{\frac{-(x-(\mu_{Z_1}-\mu_{Z_0}))^2}{2(\sigma_{Z_1}^2 + \sigma_{Z_0}^2)}} dx\\
    &= \int_{\frac{-(\mu_{Z_1}-\mu_{Z_0})}{\sqrt{\sigma_{Z_1}^2 + \sigma_{Z_0}^2}}}^\infty \frac{1}{\sqrt{2\pi}}\exp{\frac{-y^2}{2}}dy\\
    &= Q\left(\frac{\Delta G \sigma_I^2}{\sqrt{\frac2N \left[\left(2G^2 - 2G\Delta G + \Delta G^2\right)\sigma_I^4 + \left(4G-2\Delta G\right)\sigma_I^2 \sigma_n^2 + 2\sigma_n^4\right]}}\right)\\
    &= Q\left(\sqrt{\frac{N\Delta G^2\sigma_I^4}{\left(4G^2 - 4G\Delta G + 2\Delta G^2\right)\sigma_I^4 + \left(8G-4\Delta G\right)\sigma_I^2 \sigma_n^2 + 4\sigma_n^4}}\right)\\
    &= Q\left(\sqrt{\frac{N}{\left(\frac{4G^2}{\Delta G^2} - 4\frac{G}{\Delta G} + 2\right) + \frac{8G - 4\Delta G}{\Delta G^2} \frac{1}{\text{INR}} + \frac{4}{\text{INR}^2}}}\right)\\
    &= Q\left(\sqrt{\frac{N\Delta G^2\text{INR}}{\left(4G - 4G\Delta G + 2\Delta G^2\right)\text{INR} + 8G - 4\Delta G + \frac{4}{\text{INR}}}}\right)
\end{align*}
is the probability of error given $\Delta G > 0$, where $\text{INR} = \frac{\sigma_I^2}{\sigma_n^2}$. 
\end{proof}

\end{document}